\documentclass[12pt]{article}

\usepackage{color,amsmath,graphicx,amsthm,amsfonts,url}
\usepackage[displaymath]{lineno}
\usepackage{authblk}

\usepackage{caption}
\usepackage{subcaption}

\newcommand{\A}{\mathrm{A}}
\newcommand{\B}{\mathrm{B}}
\newcommand{\C}{\mathrm{C}}
\newcommand{\D}{\mathrm{D}}

\newcommand{\vs}{\mathbf{s}}
\newcommand{\vS}{\mathbf{S}}

\newcommand{\coal}{\mathrm{coal}}
\newcommand{\CRW}{\mathrm{CRW}}

\DeclareMathOperator{\Prob}{\mathbb{P}}
\DeclareMathOperator{\E}{\mathbb{E}}

\newtheorem{theorem}{Theorem}
\newtheorem{lemma}{Lemma}
\newtheorem{corollary}{Corollary}

\begin{document}

\title{Evolutionary dynamics on any population structure}

\author[1,2,3]{Benjamin Allen}
\author[4]{Gabor Lippner}
\author[5]{Yu-Ting Chen}
\author[1,6]{Babak Fotouhi}
\author[1,7]{Naghmeh Momeni}
\author[3,8]{Shing-Tung Yau}
\author[1,8,9]{Martin A.~Nowak}

 \affil[1]{Program for Evolutionary Dynamics, Harvard University, Cambridge, MA, USA}
 \affil[2] {Department of Mathematics, Emmanuel College, Boston, MA, USA}
 \affil[3] {Center for Mathematical Sciences and Applications, Harvard University, Cambridge, MA, USA}
 \affil[4] {Department of Mathematics, Northeastern University, Boston, MA, USA}
 \affil[5] {Department of Mathematics, University of Tennessee, Knoxville, TN, USA}
 \affil[6] {Institute for Quantitative Social Sciences, Harvard University, Cambridge, MA, USA}
 \affil[7] {Department of Electrical and Computer Engineering, McGill University, Montreal, Canada}
 \affil[8] {Department of Mathematics, Harvard University, Cambridge, MA, USA}
 \affil[9] {Department of Organismic and Evolutionary Biology, Harvard University, Cambridge, MA, USA}

\maketitle

\begin{abstract}
Evolution occurs in populations of reproducing individuals. The structure of a biological population affects which traits evolve \cite{ErezGraphs,NowakStructured}.  Understanding evolutionary game dynamics in structured populations is difficult. Precise results have been absent for a long time, but have recently emerged for special structures where all individuals have the same number of neighbors \cite{Ohtsuki,Taylor,chen2013sharp,allen2014games,debarre2014social}. But the problem of determining which trait is favored by selection in the natural case where the number of neighbors can vary, has remained open.  For arbitrary selection intensity, the problem is in a computational complexity class which suggests there is no efficient algorithm \cite{ibsen2015computational}. Whether there exists a simple solution for weak selection was unanswered. Here we provide, surprisingly, a general formula for weak selection that applies to any graph or social network. Our method uses coalescent theory \cite{kingman1982coalescent,WakeleyCoalescent} and relies on calculating the meeting times of random walks \cite{cox1989coalescing}. We can now evaluate large numbers of diverse and heterogeneous population structures for their propensity to favor cooperation.  We can also study how small changes in population structure---graph surgery---affect evolutionary outcomes. We find that cooperation flourishes most in societies that are based on strong pairwise ties.
\end{abstract}

Population structure affects ecological and evolutionary dynamics \cite{DurrettLevin,HassellCoexistence,NowakMay,NowakStructured}.  Evolutionary graph theory \cite{ErezGraphs,Ohtsuki,allen2014games} provides a mathematical tool for representing population structure: vertices correspond to individuals and edges indicate interactions.  Graphs can describe spatially structured populations of bacteria, plants or animals \cite{allen2013spatial}, tissue architecture and differentiation in multi-cellular organisms \cite{NowakLinear}, or social networks \cite{santos2008social,rand2014static}.  Individuals reproduce into neighboring vertices according to their fitness.  The graph topology affects the rate of genetic change \cite{allen2015molecular} and the balance of drift versus selection \cite{ErezGraphs}.  The well-mixed population, which is a classical scenario for mathematical studies of evolution, is given by the complete graph.

Of particular interest is the evolution of social behavior, which can be studied using evolutionary game theory \cite{MaynardSmith2,Hofbauer1998,broom2013game}. Evolutionary game dynamics, which are tied to ecological dynamics \cite{Hofbauer1998},
arise whenever reproductive rates are affected by interactions with others.

\begin{figure}
\begin{center}
\includegraphics[width=\textwidth]{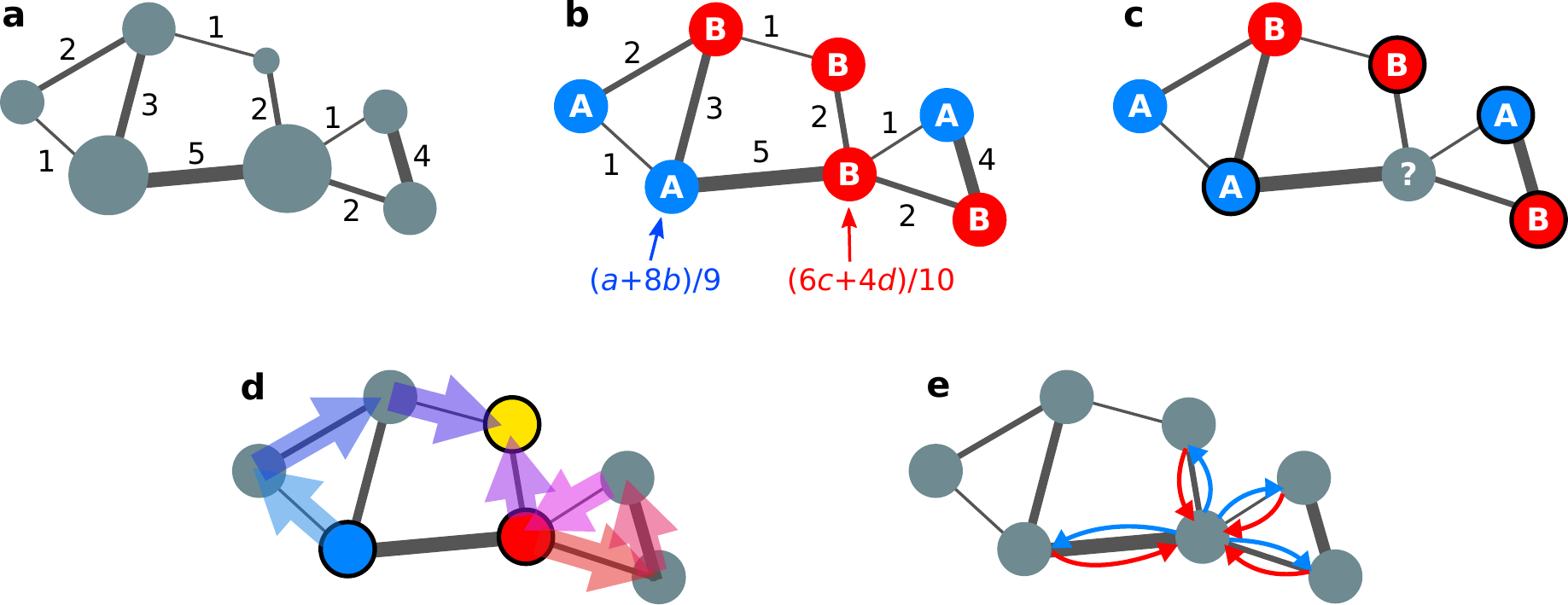}
\caption{\textbf{Evolutionary games on weighted heterogeneous graphs. a,} Population structure is represented by a  graph with edge weights $w_{ij}$.  Here vertices are sized proportionally to their weighted degree $w_i = \sum_j w_{ij}$.
\textbf{b,} Each individual $i$ interacts with each neighbor and retains the edge-weighted average payoff $f_i$, shown here for the payoff matrix \eqref{eq:game}. Payoff is translated into reproductive rate $F_i = 1 + \delta f_i$, where $\delta$ represents the strength of selection. \textbf{c,} For Death-Birth updating, first a random individual $i$ is selected to be replaced; then a neighbor $j$ is chosen with probability proportional to $w_{ij} F_{j}$ to reproduce into the vacancy.  This update rule induces competition between two-step neighbors (black circles).  \textbf{d,} Coalescent theory \cite{kingman1982coalescent,cox1989coalescing,WakeleyCoalescent,liggett2006interacting} traces ancestries back in time as random walks.  The coalescence time $\tau_{ij}$ is the expected meeting time of random walks from $i$ and $j$. The meeting point (yellow circle) represents the location of the most recent common ancestor. \textbf{e,} A key quantity is the probability $p_i$ that a randomly chosen neighbor of $i$ will chose $i$ in return, with both choices made proportionally to edge weight. Cooperation thrives when these probabilities are large.}
\label{fig:games}
\end{center}
\end{figure}  

In evolutionary games on graphs \cite{SantosScaleFree,Ohtsuki,Taylor,chen2013sharp,allen2014games,maciejewski2014evolutionary,debarre2014social}, individuals interact with neighbors according to a game and reproduce based on payoff (Fig.~\ref{fig:games}).  A central question is to determine which strategies succeed on a given graph.  It has been shown \cite{ibsen2015computational} that there cannot be a general closed-form solution or polynomial-time algorithm for this question, unless P=NP.  To make analytical progress, one can
suppose that selection is weak, meaning that the game has only a small effect on reproductive success.  In this case, exact results are known for regular graphs, where each individual has the same number of neighbors \cite{Ohtsuki,Taylor,chen2013sharp,allen2014games}. Evolutionary games on heterogenous (non-regular) graphs, which are ubiquitous in nature \cite{strogatz2001exploring}, have only been investigated using computer simulations \cite{SantosScaleFree,maciejewski2014evolutionary}, heuristic approximations \cite{konno2011condition}, and special cases \cite{Corina,hadjichrysanthou2011evolutionary,maciejewski2014evolutionary}.

Here we obtain exact results for all weighted graphs (Fig.~\ref{fig:games}a). The edge weights, $w_{ij}$, determine the frequency of game interaction and the probability of replacement between vertices $i$ and $j$. Individuals are of two types, A and B. The game is specified by a payoff matrix
\begin{equation}
\label{eq:gamemain}
\bordermatrix{ & \A & \B\cr \A & a & b\cr \B & c & d}.
\end{equation}
At each time step, each individual interacts with all of its neighbors.  The reproductive rate of individual $i$ is given by $F_i = 1+\delta f_i$, where $f_i$ is the edge-weighted average of the payoff that $i$ receives from its neighbors (Fig.~\ref{fig:games}b). The parameter $\delta$ represents the strength of selection. Weak selection is the regime $0<\delta \ll 1$.  Neutral drift, $\delta = 0$, serves as a baseline.  

As update rule, we first consider Death-Birth \cite{Ohtsuki} (Fig.~\ref{fig:games}c). An individual is chosen uniformly at random to be replaced;  then a neighbor is chosen proportionally to reproductive rate to reproduce into the vacancy.  Offspring inherit the type of their parent. Death-Birth updating is a natural scenario for genetic evolution and also translates into social settings: a random individual resolves to update its strategy; subsequently it adopts one of its neighbors' strategies proportionally to their payoff. 

Over time, the population will reach the state of all $A$ or all $B$.  Suppose we introduce a single $A$ individual at a vertex chosen uniformly at random in a population consisting of $B$ individuals.  The fixation probability, $\rho_A$, is the probability of reaching all $A$ from this initial condition. Likewise, $\rho_B$ is the probability of reaching all $B$ when starting with a single $B$ individual in a population otherwise of $A$. Selection favors $A$ over $B$ if $\rho_A>\rho_B$.

The outcome of selection depends on the spatial assortment of types, which can be studied using coalescent theory \cite{kingman1982coalescent,WakeleyCoalescent}.  Ancestral lineages are represented as random walks \cite{cox1989coalescing}.  A step from $i$ to $j$ occurs with probability $p_{ij}=w_{ij}/w_i$, where $w_i = \sum_k w_{ik}$ is the weighted degree of vertex $i$.  The coalescence time $\tau_{ij}$ is the expected meeting time of two independent random walks starting at vertices $i$ and $j$ (Fig.~\ref{fig:games}d).  Coalescence times can be obtained exactly and efficiently as the solution of the system of linear equations
\begin{equation}
\label{eq:taurecur}
\tau_{ij} = \begin{cases} 1+ \frac{1}{2} \sum_k ( p_{ik} \tau_{kj} + p_{jk} \tau_{ik} ) &  i \neq j\\
0 & i=j. \end{cases}
\end{equation} 
We show in Appendix \ref{sec:timedifferent} that the coalescence time $\tau_{ij}$ equals the expected total time during which individuals $i$ and $j$ have different types.  Therefore, if $T$ is the time to absorption (fixation or extinction of the invading type), then $T-\tau_{ij}$ is the time during which $i$ and $j$ have the same type.  Of particular interest is the expected coalescence time $t_n$ from the two ends of an $n$-step random walk, with the initial vertex chosen proportionally to weighted degree.

Our main result holds for any payoff matrix, but we first study a donation game.  Cooperators pay a cost, $c$, and provide a benefit, $b$. Defectors pay no cost and provide no benefit.  This leads to the payoff matrix
\begin{equation}
\label{eq:donation}
\bordermatrix{
& \C & \D\cr
\C & b-c & -c\cr
\D & b & 0 },
\end{equation}
For $b>c>0$, this game is a Prisoners' Dilemma.  
We find that cooperation is favored over defection, $\rho_\C>1/N>\rho_\D$, for weak selection, if and only if
\begin{equation}
\label{eq:bcT}
-c(T-t_0)+b(T-t_1)>-c(T-t_2)+b(T-t_3).
\end{equation}

Intuitively, condition \eqref{eq:bcT} states that a cooperator must have a higher average payoff than a random individual two steps away.  These two-step neighbors compete with the cooperator for opportunities to reproduce (Fig 1b). The first term, $-c(T-t_0)$, is the cost for being a cooperator, which is paid for the entire time, $T$, because $t_0=0$. The second term, $b(T-t_1)$, is the average benefit that the cooperator gets from its one-step neighbors.  For an expected time of $T-t_1$, a one-step neighbor is also a cooperator. The remaining terms, $-c(T-t_2)+b(T-t_3)$, describe the average payoff of an individual two steps away. That individual pays cost $c$ whenever it is a cooperator (time $T-t_2$) and receives benefit $b$ whenever its one-step neighbors---which are three-step neighbors of the first cooperator---are cooperators (time $T-t_3$).  

Time $T$ cancels in \eqref{eq:bcT}, leaving $-ct_2 + b(t_3-t_1)>0$.  Therefore, if $t_3>t_1$ for a given graph, cooperation is favored whenever the benefit-to-cost ratio exceeds 
\begin{equation}
\label{eq:critbc}
\left(\frac{b}{c} \right)^*=\frac{t_2}{t_3-t_1}.
\end{equation}
The critical threshold $(b/c)^*$  can be obtained exactly and efficiently for any graph by solving a system of linear equations for coalescence times and substituting into equation \eqref{eq:critbc}.  Although equation \eqref{eq:critbc} is exact only for weak selection, Monte Carlo simulations (Fig.~\ref{fig:simulation}) show that it is highly accurate for fitness costs of up to 2.5\%.  

\begin{figure}
\includegraphics[width=\textwidth]{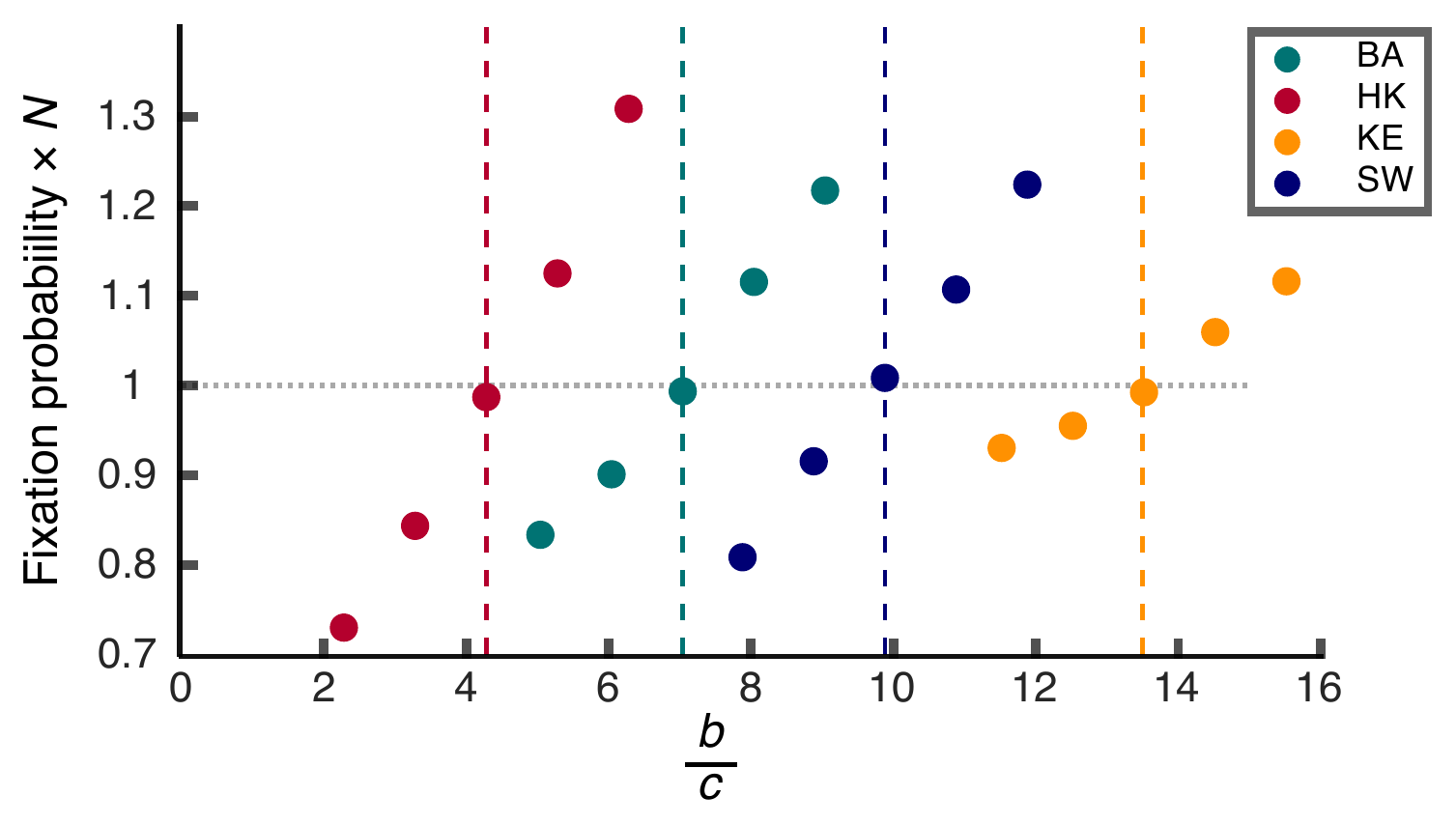}
\caption{\textbf{Simulations show accuracy of our results for moderate selection.} Our results are exact for weak selection.  To assess accuracy for nonweak selection, we performed Monte Carlo simulations with selection strength $\delta = 0.025$ and cost $c=1$.  This corresponds to a fitness cost of 2.5\%, which was empirically determined to be the cost of a cooperative behavior in yeast \cite{GoreSnowdrift}. Markers indicate population size times frequency of fixation for a particular value of $b$ on a particular graph.  Dashed lines indicate $(b/c)^*$ as calculated from Eq.~\eqref{eq:critbc}.  All graphs have size $N=100$.  Graphs are: Barabasi-Albert \cite{barabasi} (BA) with linking number $m=3$, small world \cite{newman} (SW) with initial connection distance $d=3$ and edge creation probability $p=0.025$, Klemm-Eguiluz \cite{KE} (KE) with linking number $m=5$ and deactivation parameter $\mu = 0.2$, and Holme-Kim \cite{HK} (HK) with linking number $m=2$ and triad formation parameter $P=0.2$.}
\label{fig:simulation}
\end{figure}

\begin{figure}
\includegraphics[width=\textwidth]{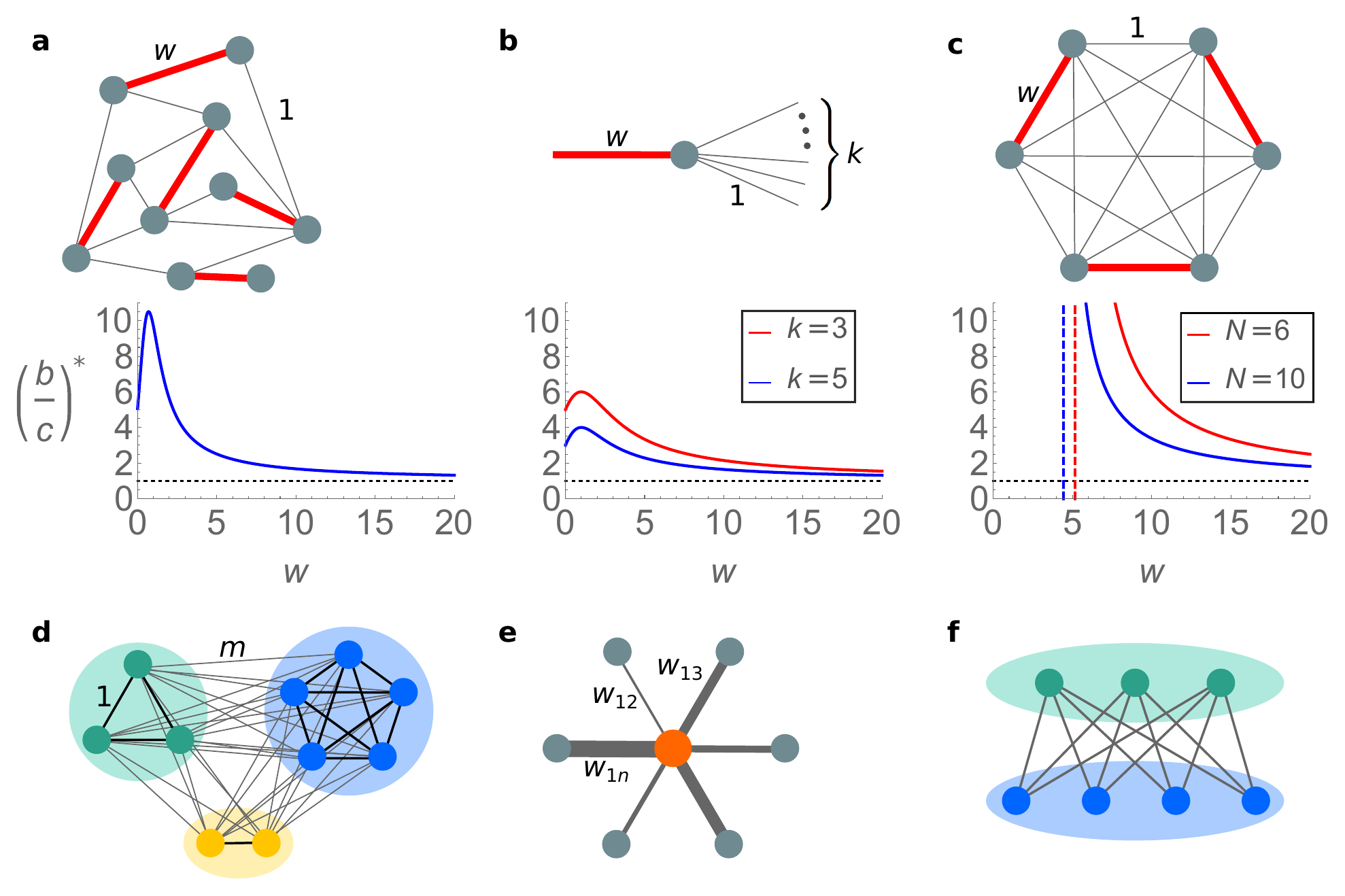}
\caption{\textbf{Graphs that promote or hinder cooperation.}  For any graph in which each individual has one partner with edge weight $w$ and all other edges have weight 1, the critical benefit-cost ratio converges to 1 as $w \to \infty$.  Thus any cooperative behavior can be favored for sufficiently large $w$. Examples include \textbf{a,} a disordered network; \textbf{b,} a weighted regular graph in which each individual has one interaction of weight $w$ and $k$ of weight 1; here $(b/c)^*=(w+k)^2/(w^2+k)$ for $N \gg k$; and \textbf{c,} a complete graph in which each individual has one partner with edge weight $w$ and all other weights are 1; here $(b/c)^*=(w-2+N)^2/[(w-2)^2-N]$;  cooperation is only possible if $w > \sqrt{N} + 2$ (dashed vertical lines). \textbf{d,} A population divided into heterogeneous islands.  Edge weights are 1 between vertices on the same island and $m<1$ between vertices on different islands.  Cooperation is most favored when the islands have equal size and migration is rare. \textbf{e,} The weighted star graph and \textbf{f,} the complete bipartite graph do not support cooperation. For both graphs, all walks of odd length are equivalent; thus $t_3=t_1$ and cooperation cannot be favored according to Eq.~\eqref{eq:critbc}.  }
\label{fig:promote}
\end{figure}

A positive value of the critical benefit-to-cost ratio means that cooperation can be favored if it is sufficiently effective. Positive values of $(b/c)^*$ always exceed---but can be arbitrarily close to---unity, at which point any cooperation that produces a net benefit is favored (Fig.~\ref{fig:promote}a--c).  A negative value, which arises for $t_3<t_1$, means that cooperation cannot be favored, but spiteful behaviors, $b<0$, $c>0$, are favored if $b/c < (b/c)^*$.  If $t_3=t_1$, then $(b/c)^*$ is infinite, and neither cooperation nor spite are favored.  

    Which networks best facilitate evolution of cooperation?  We find that cooperation thrives when there are strong pairwise ties between individuals (Fig.~\ref{fig:promote}a--c).  To quantify this property, let $p_i = \sum_j p_{ij} p_{ji}$ be the probability that a random walk from vertex $i$ returns to $i$ on its second step (Fig.~\ref{fig:games}e). In other words, $p_i$ is the probability that, if $i$ choses a neighbor (for an interaction), that choice is returned by the neighbor. Cooperation succeeds best if the $p_i$ are large.  For  graphs satisfying a particular locality property (see Appendix \ref{sec:pbar}), the critical benefit-to-cost ratio becomes $({b / c})^*={1 / \bar p}$, where $\bar{p}$ is a weighted average of the $p_i$. For unweighted regular graphs of degree $k \ll N$, a two-step walk has probability $1/k$ to return, yielding the condition \cite{Ohtsuki} $b/c>k$.

As an application, consider a population divided into islands of arbitrary sizes (Fig.~\ref{fig:promote}d).  Pairs on the same island are joined by edges of weight 1. Pairs on different islands are joined by edges of weight $m<1$. Our analytical results suggest that for a fixed number of islands, cooperation is most favored when the islands have equal size and $m$ is small (Appendix \ref{sec:island}).  

Our results also specify which population structures do not support cooperation.  For example, on a weighted star (Fig.~\ref{fig:promote}e),  random walks alternate at each step between the hub and a leaf. Any walk of odd length is equivalent to any other. Therefore $t_3=t_1$, which makes cooperation impossible.  A similar argument applies to the unweighted complete bipartite graph (Fig.~\ref{fig:promote}f).

\begin{figure}
\includegraphics[width=\textwidth]{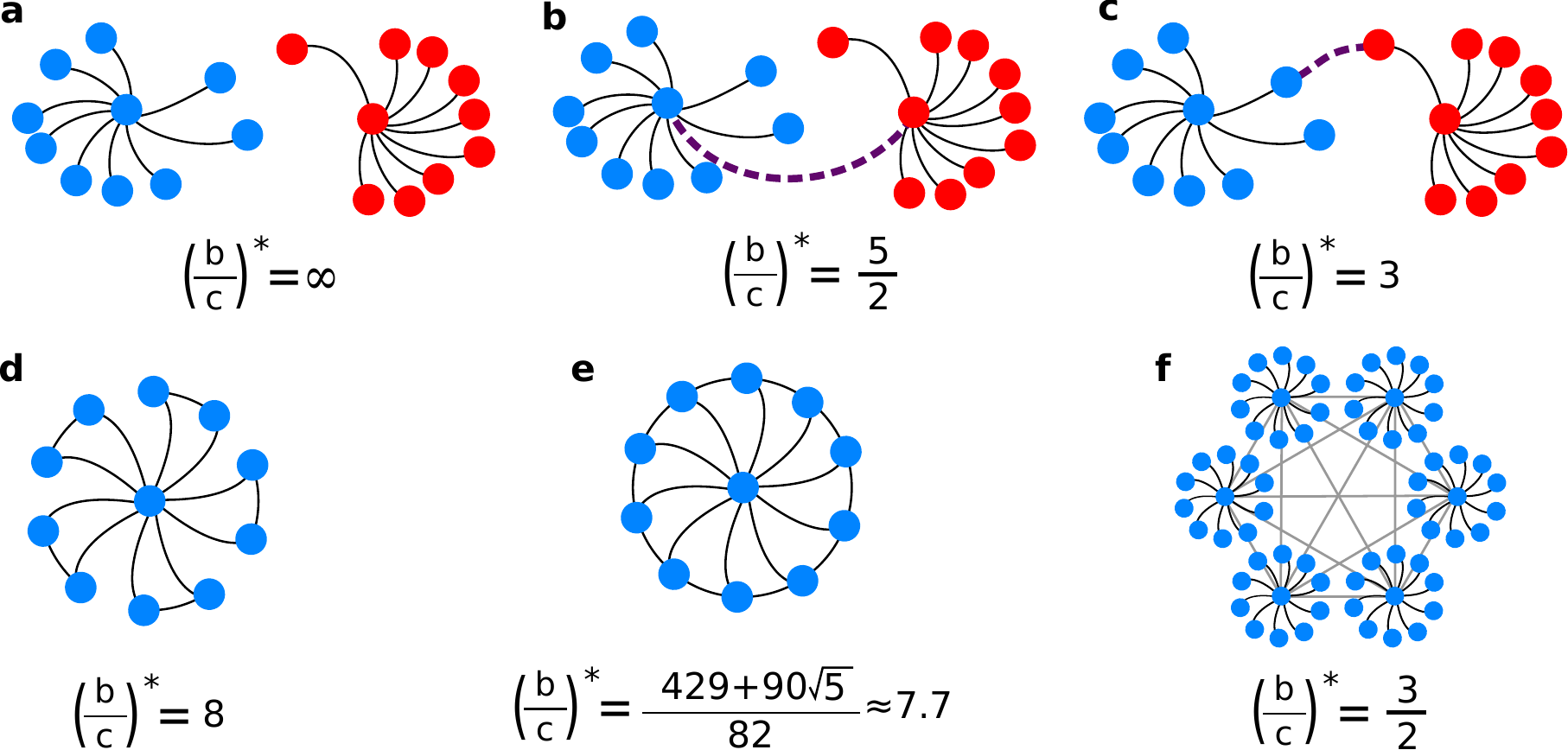}
\caption{\textbf{Rescuing cooperation by graph surgery. a,} The star does not support cooperation; the critical benefit-to-cost ratio is infinite.  \textbf{b,} Joining two stars via their hubs gives $(b/c)^*=5/2$.  \textbf{c,} Joining them via two leaves gives $(b/c)^*=3$.  \textbf{d,} A ``ceiling fan" has $(b/c)^* = 8$. \textbf{e,} A ``wheel'' has  $(b/c)^* = (429+90\sqrt{5})/82$.   \textbf{f,} If we start with $m$ stars and join their hubs in a complete graph, the critical benefit-cost ratio is $(3m-1)/(2m-2)$, which becomes $3/2$ for large $m$.  All $(b/c)^*$ values reported here hold for many leaves. Results for arbitrary sizes are given in  Appendix \ref{sec:joinedstars}.}
\label{fig:surgery}
\end{figure}

In some cases, small changes in graph topology can dramatically alter the fate of cooperation (Fig.~\ref{fig:surgery}).  Stars do not support cooperation, but joining two stars via their hubs allows cooperation for $b/c>5/2$. If we modify a star by linking pairs of leaves to obtain a ``ceiling fan", cooperation is favored for $b/c>8$.
These examples show how targeted interventions in network structure (``graph surgery") can facilitate transitions to more cooperative societies.  Of interest is also a ``dense cluster'' of stars all connected via their hubs which form a complete graph (Fig.~\ref{fig:surgery}f); for this structure the critical benefit-to-cost ratio of  $3/2$ is less than the average degree, $\bar k=2$. 

\begin{figure}
\begin{center}
\includegraphics[width=\textwidth]{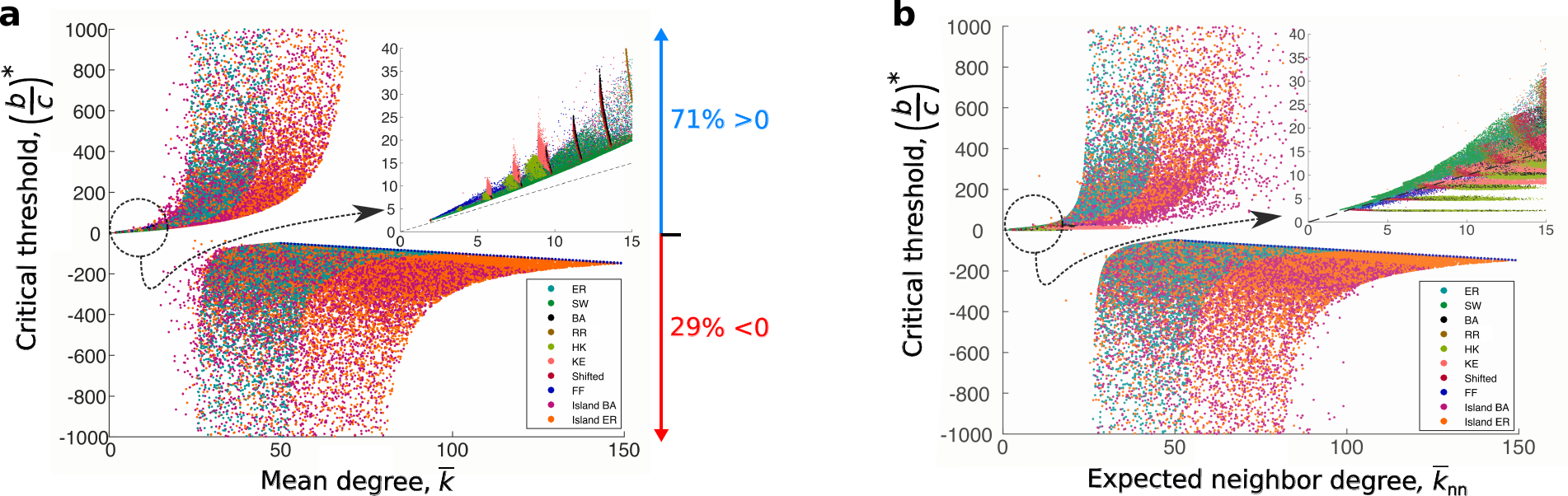}
\caption{\textbf{Conditions for cooperation on 1.3~million random graphs.}  Random graph models are Erd\"{o}s-Renyi (ER), Small World (SW), Barabasi-Albert \cite{barabasi} (BA), random recursive graph \cite{RRT} (RR), Holme-Kim \cite{HK} (HK), Klemm-Eguiliz \cite{KE} (KE), shifted-linear preferential attachment \cite{shifted} (Shifted), Forest Fire \cite{FF} (FF), and meta-networks \cite{island} of BA graphs (Island BA) and ER graphs (Island ER).  Population size $N$ varies from 100 to 150. Parameter values are given in Appendix \ref{sec:computational}.  Values of $(b/c)^*$ were obtained by solving Eq.~\eqref{eq:taurecur} and substituting into Eq.~\eqref{eq:critbc}. \textbf{a,} Scatter plot of $(b/c)^*$ versus mean degree $\bar{k}$.  71\% have positive $(b/c)^*$ and therefore admit the possibility of cooperation.  All positive $(b/c)^*$ values are larger than $\bar{k}$. Negative $(b/c)^*$ values indicate that spite ($b<0, c>0$) can be favored.  \textbf{b,} Scatter plot of $(b/c)^*$ versus $\bar{k}_\mathrm{nn}$, the expected degree of a random neighbor of a randomly chosen vertex, which has been proposed as an approximation to the critical benefit-to-cost ratio on heterogeneous graphs \cite{konno2011condition}.  Although this approximation is reasonable for many graphs, there is significant variation in $(b/c)^*$ for each value of $\bar{k}_\mathrm{nn}$.  The success of cooperation depends on features of the graph topology beyond the summary statistics $\bar{k}$ and $\bar{k}_\mathrm{nn}$.}
\label{fig:million}
\end{center}
\end{figure}  

\begin{figure}
\begin{center}
\includegraphics[width=\textwidth]{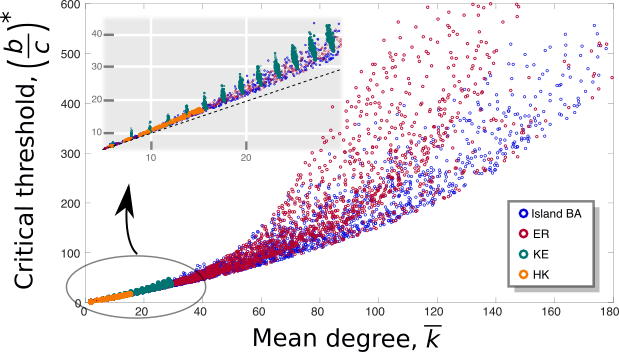}
\caption{\textbf{Conditions for cooperation on 40000 large random graphs.} Size $N$ varies from 300 to 1000.  $10^4$ graphs were generated for each of four random graph models: Erd\"{o}s-Renyi \cite{ER} (ER) with edge probability $0<p<0.25$, Klemm-Eguiluz \cite{KE} (KE) with linking number  $3\leq m \leq 5$ and deactivation parameter $0< \mu <0.15$, Holme-Kim \cite{HK} (HK) with linking number $2 \leq m \leq 4$ and triad formation parameter $0< P < 0.15$, and a meta-network \cite{island} of shifted-linear preferential attachment networks \cite{shifted} (Island BA) with $0< p_{\mathrm{inter}}<0.25$; see Methods for details.}
\label{fig:large}
\end{center}
\end{figure}  

\begin{figure}
\begin{center}
\includegraphics[width=0.75\textwidth]{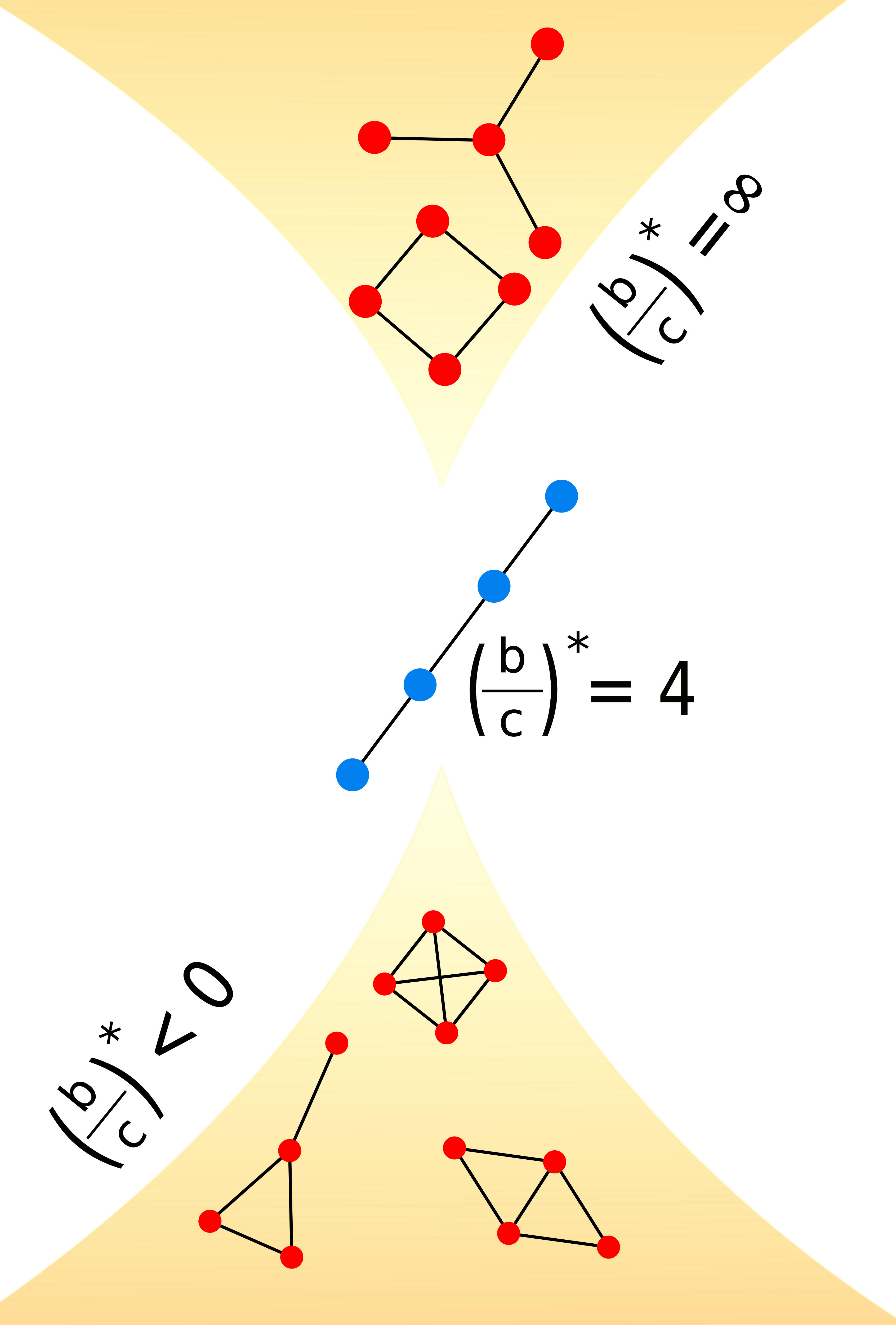}
\caption{\textbf{The critical benefit-cost threshold for all graphs of size four.}  There are six connected, unweighted graphs of size four.  Of these, only the path graph has positive $(b/c)^*$.  Two others have infinite $(b/c)^*$ and three have negative $(b/c)^*$.  There are two connected, unweighted graphs of size three (not shown): the path, which has $(b/c)^* = \infty$, and the complete graph (or triangle), which has $(b/c)^* = -2$.}
\label{fig:all4}
\end{center}
\end{figure}  

\begin{figure}
\begin{center}
\includegraphics[width=0.75\textwidth]{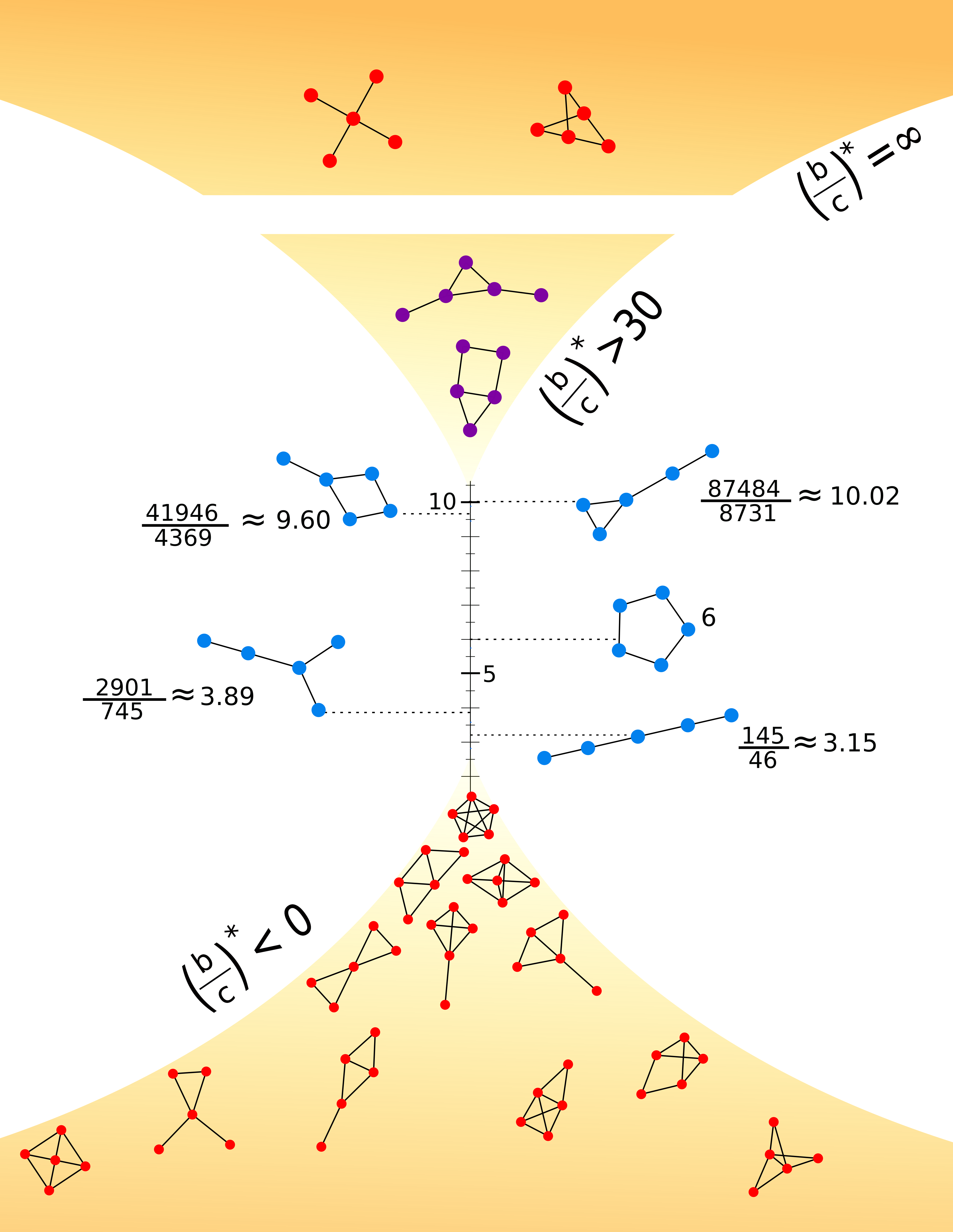}
\caption{\textbf{The critical benefit-cost threshold for all graphs of size five.}  There are 21 connected, unweighted graphs of size five.  Exact values are shown for those with $(b/c)^*$ positive and below 30.  Of the $(b/c)^*$ values, seven are positive, twelve are negative, and two are infinite.  }
\label{fig:all5}
\end{center}
\end{figure}  

\begin{figure}
\begin{center}
\includegraphics[width=0.75\textwidth]{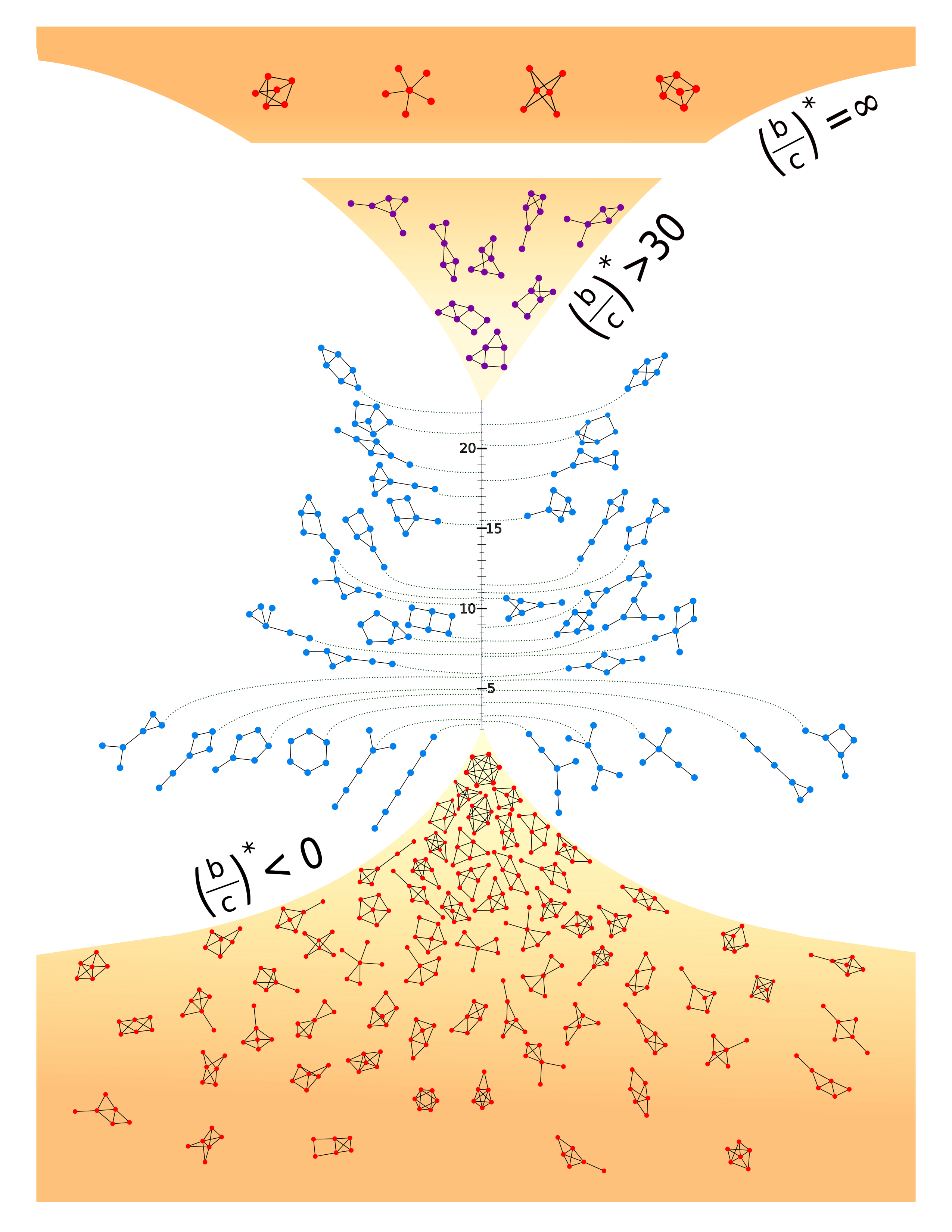}
\caption{\textbf{The critical benefit-cost threshold for all graphs of size six.}  There are 112 connected, unweighted graphs of size six.  Of these, 43 have positive $(b/c)^*$, 65 have negative $(b/c)^*$, and four have $(b/c)^*=\infty$.  Numerical values are shown for those with $(b/c)^*$ positive and below 30.  Significantly, there are graphs with the same degree sequence (for example, $3,2,2,1,1,1$) but different $(b/c)^*$.   Of the 853 graphs of size seven (not shown), 400 have positive $(b/c)^*$, 450 have negative $(b/c)^*$, and three have $(b/c)^*=\infty$. For all graphs of size up to seven, all positive values of $(b/c)^*$ are greater than the mean degree $\bar{k}$.}
\label{fig:all6}
\end{center}
\end{figure}  

\begin{figure}
\includegraphics[width=\textwidth]{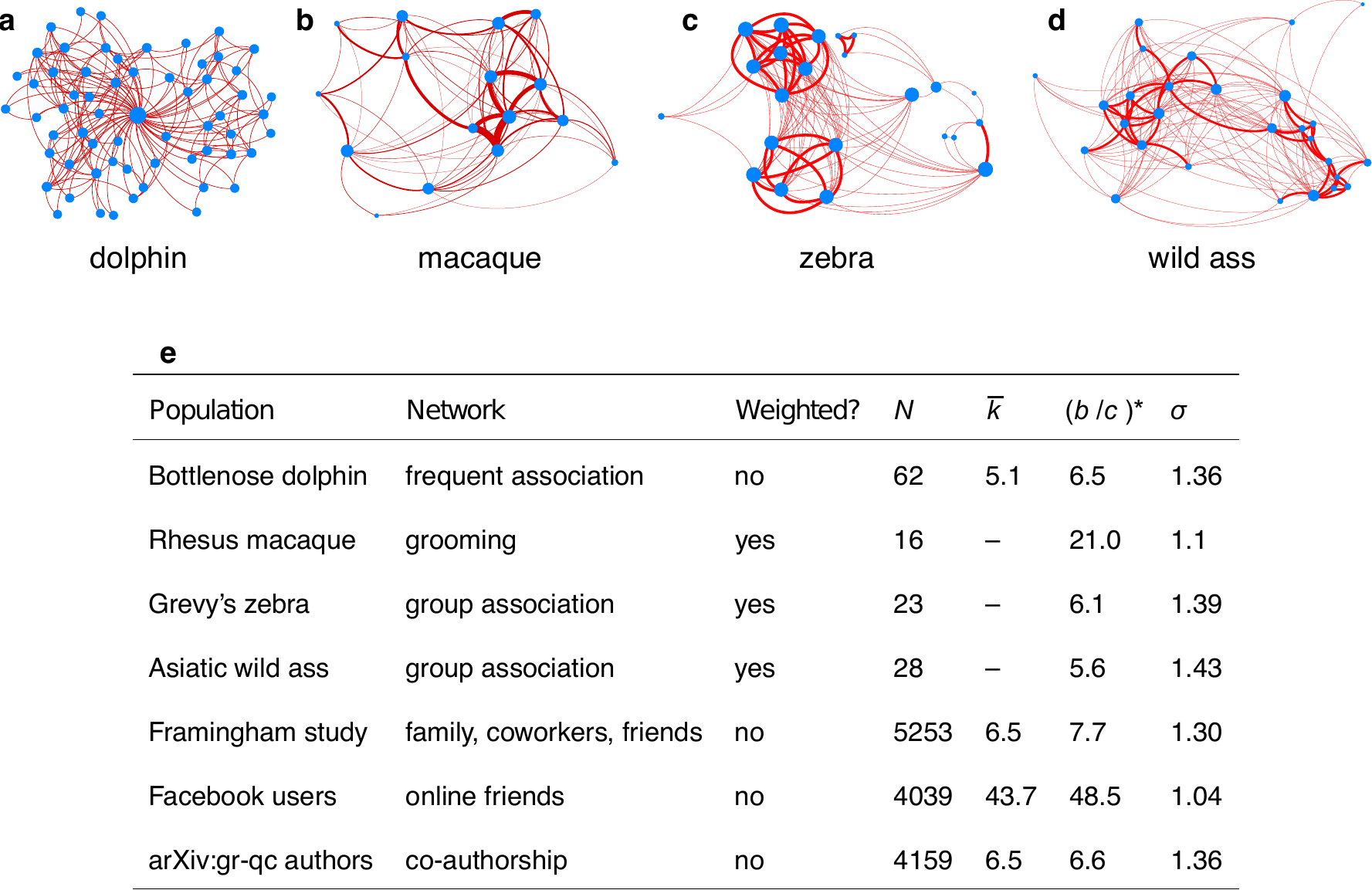}
\caption{\footnotesize{\textbf{Results for empirical networks.} The benefit-cost threshold $(b/c)^*$, or equivalently the structure coefficient \cite{Corina,NowakStructured} $\sigma$, gives the propensity of a population structure to support cooperative and/or Pareto-efficient behaviors.  These values should be interpreted in terms of specific behaviors occurring in a population, and they depend on the network ontology (that is, the meaning of links).  They can be used to facilitate comparisons across populations of similar species, or to predict consequences of changes in population structure.   \textbf{a,} Unweighted social network of frequent associations in bottlenose dolphins (\emph{Tursiops} spp.) \cite{dolphin}. \textbf{b,} Grooming interaction network in rhesus macaques (\emph{Macaca mulatta}), weighted by grooming frequency \cite{rhesus}. \textbf{c,} Weighted network of group association in Grevy's zebras (\emph{Equus grevyi}) \cite{zebra}. Preferred associations, which are statistically more frequent than random, are given weight 1.  Other associations occurring at least once are given weight $\epsilon \ll 1$.  \textbf{d,} Weighted network of group association in Asiatic wild asses (onagers) \cite{zebra}, with the same weighting scheme as for the zebra network. For both zebra and wild ass, the unweighted networks, including every association ever observed, are dense and behave like well-mixed populations. In contrast, the weighted networks, which emphasize close ties, can promote cooperation.  \textbf{e,} Table showing data from panels a--d as well as a social network of family, self-reported friends, and coworkers as of 1971 from the Framingham Heart Study \cite{hill_infectious_2010,christakis2007spread}, a Facebook ego-network \cite{mcauley2012learning},  and the co-authorship network for the General Relativity and Quantum Cosmology category of the arXiv preprint server \cite{leskovec2007graph}. Average degree is reported for unweighted graphs only; for weighted graphs it is unclear which notion of degree is most relevant.  Note that large  $(b/c)^*$ ratios, which correspond to $\sigma$ values very close to one, do not mean that cooperation is never favored. Rather, the implication is that the network behaves similarly to a large well-mixed population, in which cooperation is favored for any $2 \times 2$ game with $a+b>c+d$ .  The donation game does not satisfy this inequality, but other cooperative interactions do \cite{hauert2006synergy,nowak2012evolving}.}}
\label{fig:animal}
\end{figure}

To explore the evolutionary consequences of graph topology on a large scale, we calculated $(b/c)^*$ four four ensembles of graphs: (i) 1.3 million unweighted graphs of sizes $100\leq N \leq 150$, generated by ten  random graph models (Fig.~\ref{fig:million}), (ii) 40K unweighted graphs of sizes $300\leq N \leq 1000$ generated by four random graph models (Fig.~\ref{fig:large}), (iii) every unweighted graph of size up to seven (Fig.~\ref{fig:all4}--\ref{fig:all6}), and (iv) seven empirical human and animal social networks (Fig.~\ref{fig:animal}).  In general we find that, as the average degree, $\bar{k}$,  increases, cooperation becomes increasingly difficult and eventually impossible.  However, there is considerable variance in $(b/c)^*$ for each value of $\bar{k}$.  The success of cooperation is not determined by the average degree, or even the entire degree sequence (Fig.~\ref{fig:all6}).

\begin{figure}
\begin{center}
\includegraphics[scale=0.8]{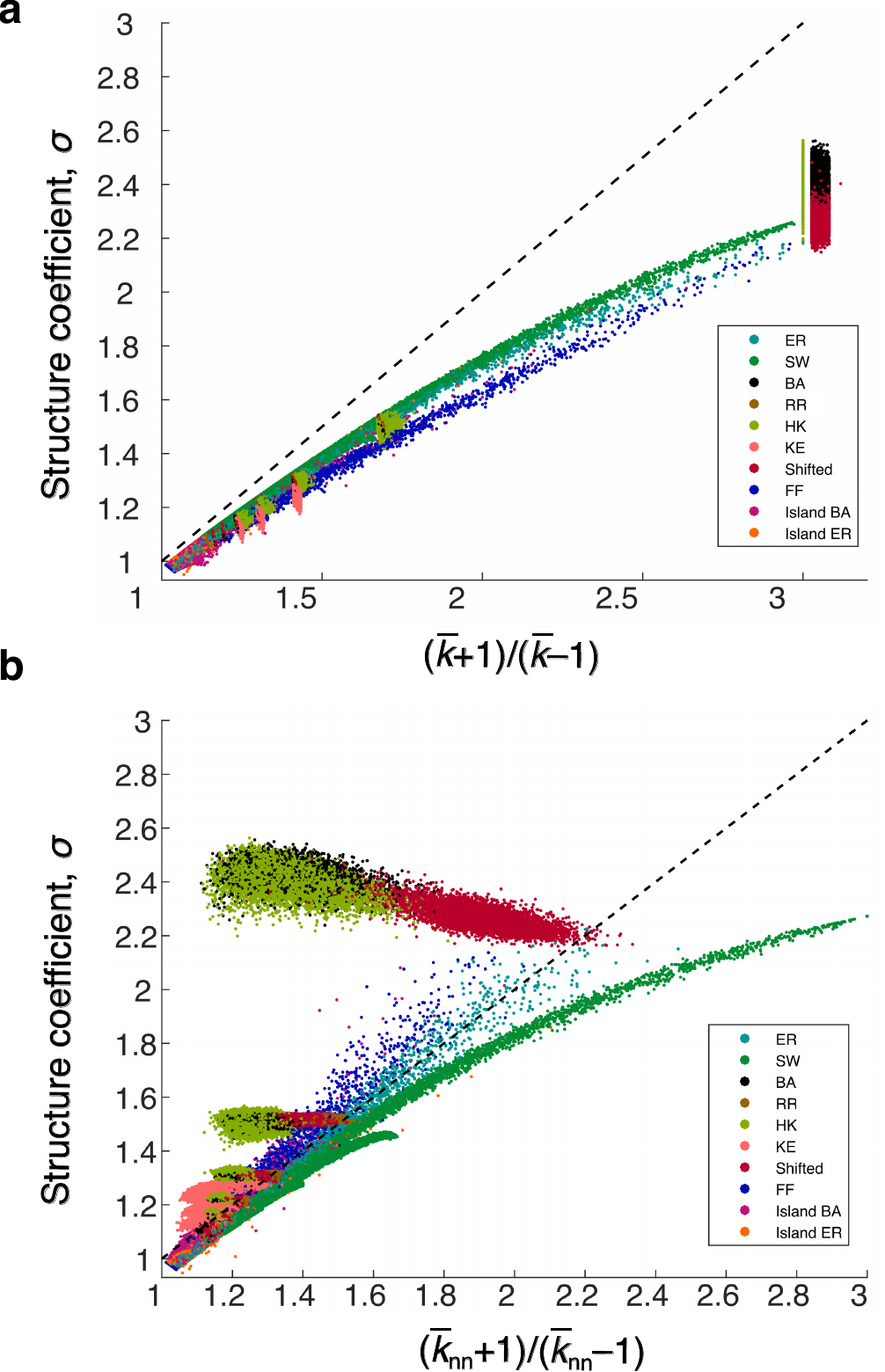}
\caption{\textbf{Structure coefficients for 1.3~million random graphs.}  We computed the structure coefficient \cite{Corina} $\sigma = [(b/c)^*+1]/[(b/c)^*-1]$ for the same ensemble of random graphs as in Fig.~\ref{fig:million} of the main text.  Strategy $A$ is favored over strategy $B$ under weak selection if $\sigma a + b > c + \sigma d$; see Eq.~\eqref{eq:game} of Methods.  \textbf{a,} Scatter plot of $\sigma$ versus $(\bar{k}+1)/(\bar{k}-1)$, which is the $\sigma$-value for a regular graph of the same mean degree $\bar{k}$.  \textbf{b,} Scatter plot of $\sigma$ versus $(\bar{k}_\mathrm{nn}+1)/(\bar{k}_\mathrm{nn}-1)$, which is the $\sigma$-value one would expect if the condition \cite{konno2011condition} $b/c > \bar{k}_\mathrm{nn}$ were exact.  Here, $\bar{k}_\mathrm{nn}$ is the expected degree of a neighbor of a randomly chosen vertex.}
\label{fig:sigmamillion}
\end{center}
\end{figure} 

So far we have discussed the donation game \eqref{eq:donation}, but our theory extends to any pairwise game interaction of the form \eqref{eq:game}.  On any graph, the condition for natural selection to favor strategy A over strategy B, in the sense that $\rho_\A>\rho_\B$ for weak selection, can be written \cite{Corina} as $\sigma a +b > c+ \sigma d$. Our result implies that $\sigma = (-t_1+t_2+t_3)/ (t_1+t_2-t_3)$.  Therefore, the key quantity $\sigma$ can be calculated for any graph by solving a system of linear equations (Fig.~\ref{fig:sigmamillion}).  The value of $\sigma$ quantifies the extent to which a graph supports cooperation in a social dilemma, or the Pareto-efficient equilibrium in a coordination game \cite{Corina,NowakStructured}.

Our model can be extended in various ways.  Birth-Death updating \cite{Ohtsuki} can be studied. Total instead of average payoff can be used to compute reproductive rates \cite{maciejewski2014evolutionary}.  Different graphs can specify interaction and replacement \cite{OhtsukiBreaking,Taylor,allen2014games}.  Mutation can be introduced \cite{AllenGraphMut,allen2014games}. For each of these variations, we obtain the exact critical benefit-to-cost ratios (and $\sigma$ values) in terms of coalescence times (Appendix \ref{sec:variations}).  Additionally, appendix \ref{sec:DirectInclusive} provides interpretations of our results in terms of direct fitness and---in the special case of the donation game \eqref{eq:donation}---inclusive fitness.  

In summary, we report here the first analytic result describing how natural selection
chooses between competing strategies on any graph for weak selection. Our framework applies to strategic interactions among humans in social
networks, as well as ecological and evolutionary interactions among simpler organisms
in any kind of spatially structured population (Fig.~\ref{fig:animal}). Our results reveal which population structures promote certain behaviors, such as cooperation. We find that cooperation
flourishes most in the presence of strong pairwise ties, which is an intriguing 
mathematical argument for the importance of stable partnerships in forming
the backbone of cooperative societies.  

\section*{Acknowledgements} Supported by a grant from the John Templeton Foundation.  The Program for Evolutionary Dynamics is supported in part by a gift from B.~Wu and Eric Larson.  We are grateful to Siva R.~Sundaresan and Daniel I.~Rubenstein for providing us with data on zebra and wild ass networks.

\appendix

\section{Model and notation}

\subsection{Graph}

Population structure is represented by a weighted, connected graph $G$ with edge weights $w_{ij}, \; i,j \in G$. $G$ is undirected: $w_{ji} = w_{ij}$ for each $i,j \in G$.  Self-loops are allowed, and are represented by the weights $w_{ii}$.  Self-loops represent interaction with oneself and replacement by one's own offspring.  

We define the \emph{weighted degree} of vertex $i$ as $w_i=\sum_{j \in G} w_{ij}$.  The total sum of all edge weights (counting both directions for each pair of vertices) is denoted $W$:
\[
W = \sum_{i,j \in G} w_{ij} = \sum_{i \in G} w_i.
\]

\subsection{Random walks}

We will make use of random walks on $G$ in both discrete and continuous time.  For a random walk on $G$, steps are taken with probability proportional to edge weight; thus the probability of a step from $i$ to $j$ is $p_{ij} = w_{ij}/w_i$.  The probability that an $n$-step walk from vertex $i$ terminates at vertex $j$ is denoted $p_{ij}^{(n)}$.  

There is a unique stationary distribution $\{\pi_i\}_{i \in G}$ for random walks on $G$, in which vertex $i$ has stationary probability $\pi_i = w_i/W$.  This means that for each $i,j \in G$, $\lim_{n \to \infty} p_{ij}^{(n)} = \pi_j$.  Random walks have the reversibility property that for each $i,j \in G$, $\pi_i p_{ij}^{(n)} = \pi_j p_{ji}^{(n)}$.

We will use the following shorthand: for any function $g_i$ on $G$, we define 
\[
g_i^{(n)} = \sum_{j \in G} p_{ij}^{(n)} g_j.
\]
That is, $g_i^{(n)}$ is the expected value of $g_j$ where $j$ is the terminus of an $n$-step random walk from $i$.  For any function $h_{ij}$ on $G \times G$, we define 
\[
h^{(n)} = \sum_{i,j \in G} \pi_i p_{ij}^{(n)} h_{ij}.
\]
In words, $h^{(n)}$ is the expected value of $h_{ij}$ where $i$ and $j$ are the two ends of an $n$-step random walk on $G$ started from the stationary distribution.

\subsection{Evolutionary Markov chain}

For most of this Supplement we consider a continuous-time version of the Death-Birth (DB) process.  The translation from discrete to continuous time does not affect fixation probabilities or other metrics for evolutionary success.

The state of the process is represented as a binary vector $\vs = (s_i)_{i \in G} \in \{0,1\}^G$, where 0 and 1 correspond to the two types and $s_i \in \{0,1\}$ denotes the type of vertex $i$.  

Evolution is modeled as a continuous-time Markov chain $( \vS(t))_{t \geq 0}$ on $\{0,1\}^G$, which we call the \emph{evolutionary Markov chain}.  The dynamics of the evolutionary Markov chain are described in the following subsections.

\subsection{Payoff and reproductive rate}

The edge-weighted average payoff to vertex $i$ in state $\vs$ is denoted $f_i(\vs)$.  We first consider the general game
\begin{equation}
\label{eq:game}
\bordermatrix{ & \mathrm{A} & \mathrm{B}\cr \mathrm{A} & a & b\cr \mathrm{B} & c & d},
\end{equation}
Letting 1 correspond to A and 0 to B, we have
\begin{equation}
\label{eq:fgeneral}
f_i(\vs) = a s_i s_i^{(1)} + b s_i \left(1-s_i^{(1)} \right) 
+ c (1-s_i) s_i^{(1)} + d (1-s_i)  \left(1-s_i^{(1)} \right).
\end{equation}
For the donation game
\begin{equation}
\label{eq:PD}
\bordermatrix{
& \C & \D\cr
\C & b-c & -c\cr
\D & b & 0 },
\end{equation}
letting 1 correspond to C and 0 to D, Eq.~\eqref{eq:fgeneral} reduces to
\begin{equation}
\label{eq:fPD}
f_i(\vs) = -cs_i + b \sum_{j \in G} p_{ij} s_j.
\end{equation}

The reproductive rate of vertex $i$ in state $\vs$ is $F_i(\vs) = 1 + \delta f_i(\vs)$, where $\delta > 0$ quantifies the strength of selection.  

\subsection{Transitions}

State transitions in the evolutionary Markov chain occur via \emph{replacement events}, in which the occupant of one vertex is replaced by the offspring of another.  We denote by $i \rightarrow j$ the event that the occupant of $j \in G$ is replaced by the offspring of $i \in G$.  Replacement events occur as Poisson processes, with rates depending on the state $\vs$.  The Death-Birth process is defined using the rates
\begin{equation}
\label{eq:DBrate}
\operatorname{Rate}[i \rightarrow j](\vs) = \frac{w_{ij} F_i(\vs)}{\sum_{k \in G} w_{kj} F_k(\vs)}.
\end{equation}
According to Eq.~\eqref{eq:DBrate}, each vertex is replaced at overall rate 1.  The conditional probability that $i$ reproduces, given that vertex $j$ is replaced, is proportional to $w_{ij} F_i(\vs)$.  In this way, Eq.~\eqref{eq:DBrate} defines a continuous-time analogue of the DB process described in the main text.

If the replacement event $i \rightarrow j$ occurs in state $\vs$, and $s_j \neq s_i$, then a state transition occurs and the new state $\vs'$ is defined by $s_j'=s_i$ and $s_k'=s_k$ for all $k \neq j$.  (That is, vertex $j$ inherits the type of vertex $i$, and all other vertices retain their type.)  If $s_j = s_i$, then no transition occurs and $\vs'=\vs$.

\section{Fixation probability under weak selection}
\label{sec:fixweak}

The evolutionary Markov chain has two absorbing states: the state $\mathbf{1}$ for which $s_i = 1$ for all $i \in G$, and the state $\mathbf{0}$ for which $s_i=0$ for all $i \in G$.  These states correspond to the fixation of types B and A, respectively.  All other states of the evolutionary Markov chain are transient \cite[Theorem 2]{allen2014measures}.  Thus from any given initial state, the evolutionary Markov chain will eventually become absorbed in one of these two fixation states.  We denote by $\rho_{\vs_0}$ the fixation probability of type A from state $\vs_0 \in \{0,1\}^G$---that is, the probability that, from initial state $\vs_0$, the evolutionary Markov chain becomes absorbed in state $\mathbf{1}$. 

We are interested in the behavior of these fixation probabilities under weak selection---that is, to first order in $\delta$ as $\delta \to 0^+$.  Chen \cite{chen2013sharp} derived a weak-selection perturbation formula for $\rho_{\vs_0}$, which applies to the birth-death and death-birth processes considered here.  To provide intuition for this formula, we give a heuristic derivation here, referring to \cite{chen2013sharp} for mathematical details.

Our analysis focuses on the \emph{degree-weighted frequency} of type A:
\[
\hat{s} = \sum_{i \in G} \pi_i s_i.
\]
The degree weighted frequency at time $t$ is represented by the random variable
\[
\hat{S}(t) = \sum_{i \in G} \pi_i S_i(t).
\]
The main idea is that $\hat{S}(t)$ is a Martingale for the neutral process ($\delta=0$); and that weak selection can be understood as a perturbation of this Martingale.  The weighting $\pi_i$ of vertex $i$ can be understood as its reproductive value \cite{taylor1990allele,maciejewski2014reproductive}.  Similar arguments have been used in other contexts \cite{TaylorHow,leturque2002dispersal,RoussetBook,LessardFixation,tarnita2014measures,van2015social}.

Consider the evolutionary Markov chain with arbitrary initial state $\vS(0)=\vs_0 \in \{0,1\}^G$.  By the Fundamental Theorem of Calculus, the expected degree-weighted frequency $\E_{\vs_0} [\hat{S}(T)]$ at time $T>0$ satisfies
\begin{equation}
\E_{\vs_0} [\hat{S}(T) ] = \hat{s}_0 + \int_0^T \frac{d}{dt} \E_{\vs_0} [\hat{S}(t)] \; dt.
\end{equation}
In the limit $T \to \infty$, the expected degree-weighted frequency of type 1 becomes equal to its fixation probability; therefore we have
\begin{equation}
\label{eq:rhointegral}
\rho_{\vs_0} = \hat{s}_0 + \int_0^\infty \frac{d}{dt} \E_{\vs_0}[\hat{S}(t)] \; dt.
\end{equation}

We now define a state function $D(\vs)$ giving the expected instantaneous rate of change in the degree-weighted frequency of type A from state $\vs$.  $D(\vs)$ is defined by the relation
\begin{equation}
\label{eq:Ddef}
\E \left[\hat{S}(t+\epsilon)-\hat{S}(t) \big| \hat{S}(t)=\vs \right] \; = \;  D(\vs) \epsilon + o(\epsilon) \qquad (\epsilon \to 0^+).
\end{equation}
The exact form of $D(\vs)$ for DB updating is derived in Section \ref{sec:shortterm}.  Substituting in Eq.~\eqref{eq:rhointegral}, we obtain
\begin{equation}
\label{eq:rhoD}
\rho_{\vs_0}= \hat{s}_0 + \int_0^\infty \E_{\vs_0} [D(\vS(t))] \; dt
\end{equation}

We now consider the case $\delta = 0$, which represents neutral drift (no selection).  We indicate neutral drift by the superscript ${}^\circ$.  We will show in Section \ref{sec:shortterm} that, under DB updating, $D^\circ(\vs)=0$ for all $\vs \in \{0,1\}^G$, meaning that $\hat{S}(t)$ is a Martingale for neutral drift. It follows from Eq.~\eqref{eq:rhoD} that $\rho_{\vs_0}^\circ = \hat{s}_0$; that is, the fixation probability of type A, under neutral drift, is equal to its initial degree-weighted frequency (i.e.,~its initial reproductive value \cite{taylor1990allele,maciejewski2014reproductive}).  In particular, a neutral mutation arising at vertex $i$ has fixation probability $\pi_i$.

We now turn to weak selection; that is, we consider asymptotic expansions as $\delta \to 0^+$.  Since $D^\circ(\vs)=0$, we have
\begin{equation}
\label{eq:Dweak}
D(\vs) = \delta D'(\vs) + \mathcal{O}(\delta^2).
\end{equation}

This allows us to expand the integrand in Eq.~\eqref{eq:rhointegral}:
\begin{align}
\nonumber
\E_{\vs_0}[D(\vS(t))] & = \sum_\vs \Prob_{\vs_0}[\vS(t)=\vs] D(\vs)\\
\nonumber
& = \delta \sum_\vs \Prob_{\vs_0}^\circ[\vS(t)=\vs] D'(\vs) +\mathcal{O}(\delta^2) \\
\label{eq:EDexpand}
& = \delta \E_{\vs_0}^\circ[D'(\vS(t))] +\mathcal{O}(\delta^2).
\end{align}

If we could freely interchange the expansion \eqref{eq:EDexpand} with the integral in Eq.~\eqref{eq:rhointegral}, we would have
\begin{equation}
\label{eq:rhoweak}
\rho_{\vs_0}= \hat{s}_0 + \delta \int_0^\infty \E_{\vs_0}^\circ [D'(\vS(t))] \; dt + \mathcal{O}(\delta^2).
\end{equation}
A formal justification for this interchange and proof of Eq.~\eqref{eq:rhoweak}, for a class of models that includes the DB process considered here, is given in Theorem 3.8 of Chen \cite{chen2013sharp}.  

For convenience, we introduce the following notation: for any function of the state $g(\vs)$ and any initial state $\vs_0$, we define
\begin{equation}
\langle g \rangle^\circ_{\vs_0} = \int_0^\infty \E_{\vs_0}^\circ [g(\vS(t))] \; dt.
\end{equation}
Thus Eq.~\eqref{eq:rhoweak} can be rewritten as
\begin{equation}
\label{eq:rhoweak2}
\rho_{\vs_0}= \hat{s}_0 + \delta \langle D' \rangle^\circ_{\vs_0} + \mathcal{O}(\delta^2).
\end{equation}

Of particular biological importance is case that type 1 initially occupies only a single vertex.  In biological terms, a single indivdual of a new type has appeared in the population, e.g.~by mutation or migration.  We therefore focus in particular on initial states $\vs_0$ with exactly one vertex of type 1.  Let $\mathbf{u}$ be the probability distribution over states that assigns probability $1/N$ to all states $\vs$ with exactly one vertex of type 1, and probability zero to all other states.  (That is, we suppose that the new type is equally likely to arise at each vertex.  This is a natural assumption for DB updating, but not necessarily for other update rules \cite{allen2014measures,tarnita2014measures,allen2015molecular}.)  We use the subscript $\mathbf{u}$ to denote the expected value of a quantity when the initial state of the evolutionary Markov chain is sampled from $\mathbf{u}$.  

We define the overall fixation probability of type A as $\rho_\mathrm{A}=\rho_{\mathbf{u}}$, the probability that A becomes fixed when starting from a single vertex chosen with uniform probability.  Taking the expectation of Eq.~\eqref{eq:rhoweak2} with $\vs_0$ sampled from $\mathbf{u}$, we have
\begin{equation}
\label{eq:rhoweaku1}
\rho_\mathrm{A}= \frac{1}{N} + \delta \langle D' \rangle^\circ_{\mathbf{u}} + \mathcal{O}(\delta^2).
\end{equation}

\section{Instantaneous change under weak selection}
\label{sec:shortterm}

Here we compute the expected instantaneous rate of degree-weighted frequency change $D(\vs)$ from a state $\vs$, under DB updating and weak selection.  Note that if the event $i \to j$ occurs (that is, if $j$ is replaced by the offspring of $i$), the resulting change in $\hat{s}$ is $\pi_j(s_i-s_j)$.  Thus the expected instantaneous rate of degree-weighted frequency change is given by
\begin{align*}
D(\vs) & =  \sum_{j \in G} \pi_j\left( -s_j + \sum_{i \in G} s_i \frac{w_{ij} F_i(\vs)}{\sum_{k \in G} w_{kj} F_k(\vs)}\right)\\
& = \sum_{i \in G} s_i \left( - \pi_i + \sum_{j \in G}  \pi_j \frac{w_{ij} F_i(\vs)}{\sum_{k \in G} w_{kj} F_k(\vs)}\right) \\
& = \delta \sum_{i \in G}  s_i \left( \pi_i f_i(\vs)  - \pi_i\sum_{k \in G} p_{ik}^{(2)} f_k(\vs) \right) + \mathcal{O}(\delta^2)\\
& = \delta \sum_{i \in G} \pi_i  s_i \left( f_i^{(0)}(\vs) - f_i^{(2)}(\vs) \right) + \mathcal{O}(\delta^2).
\end{align*}
This shows that $D^\circ(\vs)=0$ for all states $\vs$ (thus $\hat{S}(t)$ is a Martingale as claimed earlier), and
\begin{equation}
\label{eq:D's}
D'(\vs) = \sum_{i \in G} \pi_i  s_i \left( f_i^{(0)}(\vs) - f_i^{(2)}(\vs) \right).
\end{equation}

For the special case of the donation game \eqref{eq:PD}, we have
\begin{equation}
\label{eq:D'sPD}
D'(\vs) = \sum_{i \in G} \pi_i s_i \left( -c \left(s_i^{(0)}-s_i^{(2)} \right) + b \left(s_i^{(1)} - s_i^{(3)} \right) \right).
\end{equation}
Applying Eq.~\eqref{eq:rhoweaku1}, the fixation probability of cooperation is given by 
\begin{equation}
\label{eq:sfix}
\rho_\C = \frac{1}{N} +  
\delta \sum_{i \in G} \pi_i \left ( - c \left \langle s_i \left(s_i^{(0)} - s_i^{(2)} \right) \right \rangle^\circ_{\mathbf{u}}
+ b \left \langle s_i \left(s_i^{(1)} - s_i^{(3)} \right)\right \rangle^\circ_{\mathbf{u}} \right)
+ \mathcal{O}(\delta^2).
\end{equation}
In the following section we will show how the quantities in Eq.~\eqref{eq:sfix} can be computed using coalescence times.

\section{Coalescing random walks}
\label{sec:CRW}

A \emph{coalescing random walk (CRW)} \cite{cox1989coalescing,liggett2006interacting} is a collection of random walks on $G$ that step independently until two walks meet (or ``coalesce"), after which these two step together.  This process models the tracing of ancestors backwards in time.  We will consider both continuous-time and discrete-time versions of the coalescing random walk, starting in either case with two walkers.

\subsection{Continuous-time and discrete-time CRWs}

In the continuous-time version, we consider a pair of walkers $(X(t), Y(t))_{t \geq 0}$ with arbitrary initial vertices  $X(0)=i$ and $Y(0)=j$.  Each steps at Poisson rate $1$, corresponding to the rate at which sites are replaced in the continuous-time Death-Birth process.  $X(t)$ and $Y(t)$ step independently until their time of coalescence (first meeting), which is denoted $T_\coal$.  After this time, $X(t)$ and $Y(t)$ step together, so that $X(t)=Y(t)$ for all $t > T_\coal$. Probabilities and expectations under the continuous-time coalescing random walk from $i$ and $j$ are denoted $\Prob_{(i,j)}^\mathrm{CRW}[\;]$ and $\E_{(i,j)}^\mathrm{CRW}[\;]$ respectively.

In the discrete-time version, we consider a pair of walkers $(X(t), Y(t))_{t =0}^\infty$, again with arbitrary initial vertices  $X(0)=i$ and $Y(0)=j$. At each time step $t=0, 1, \ldots$, if $X(t) \neq Y(t)$, one of the two walkers is chosen, with equal probability, to take a random walk step.  If $X(t) = Y(t)$, then both make the same random walk step.  The coalescence time $T_\coal$ is defined as the first time $t$ for which $X(t)=Y(t)$. We use tildes to indicate the discrete-time process, so that probabilities and expectations under the discrete-time coalescing random walk from $i$ and $j$ are denoted $\tilde{\Prob}_{(i,j)}^\mathrm{CRW}[\;]$ and $\tilde{\E}_{(i,j)}^\mathrm{CRW}[\;]$ respectively.

We note that there is a difference in time-scales for the discrete-time and continuous-time CRWs. In the continuous-time CRW, two steps occur per unit time on average (since each walker steps at rate 1). In the discrete-time CRW, one step in total is taken per unit time.

\subsection{Coalescence times}
\label{sec:coalescencetimes}

We denote the expected coalescence time from $i$ and $j$ in the discrete-time coalescing random walk by $\tau_{ij} = \tilde{\E}^\mathrm{CRW}_{(i,j)}[T_\mathrm{coal}]$.  Because coalescence times can be understood as hitting times to the diagonal of $G \times G$, and because expected hitting times are preserved under transitions between discrete and continuous time 
\cite[\S 2.5.3]{aldous2002reversible}, we have $\tau_{ij} = 2 \E^\mathrm{CRW}_{(i,j)}[T_\coal]$.  (The factor of two is due to the difference in time-scales.)

Now suppose that $i$ and $j$ are the two ends of a random walk of length $n$ started from the stationary distribution.  Taking the expectation of $\tau_{ij}$ over all such choices, we obtain the quantity $\tau^{(n)}$, which is denoted  denoted $t_n$ in the main text:
\begin{equation}
t_n = \tau^{(n)} = \sum_{i,j \in G} \pi_i p_{ij}^{(n)} \tau_{ij}.
\end{equation}

\subsection{Assortment and coalescence}
\label{sec:assortment}

Coalescing random walks represent the ancestry of a set of individuals traced backwards in time, and can therefore be used to study assortment.  Mathematically, coalescing random walks are dual to the neutral ($\delta=0$) case of our model.  This duality implies that, for any initial state $\vs_0$ and any pair of types $x,y \in \{0,1\}$,
\begin{equation}
\label{eq:voterduality}
\Prob^\circ_{\vs_0} \big[s_i(t)=x,s_j(t)=y \big] 
\; = \; \Prob^\CRW_{(i,j)} \left[(\vs_0)_{X(t)} = x, (\vs_0)_{Y(t)} = y \right].
\end{equation}
In biological language, this relation says that the current occupants of vertices $i$ and $j$ have the same types as their corresponding ancestors in the initial state.  

We first apply this result to an initial state $\vs_0$ that has a single vertex $k$ of type 1 and all others of type 0: $(\vs_0)_k = 1$ and $(\vs_0)_\ell = 0$ for all $\ell \neq k$.  Consider the random variable $S_i(t) \, S_j(t)$, which equals one if vertices $i$ and $j$ both have type 1 at time $t$ and zero otherwise.  Applying the duality relation \eqref{eq:voterduality}, we find
\begin{align}
\nonumber
\E^\circ_{\vs_0} \big[S_i(t) \, S_j(t) \big] 
& = \Prob^\CRW_{(i,j)} \left[X(t) =Y(t) = k \right]\\
\label{eq:coalescek}
& = \Prob^\CRW_{(i,j)} \left[T_\coal < t,  X(t) = k \right].
\end{align}
That is, vertices $i$ and $j$ both have type 1 at time $t$ if and only if they are both descended from the original vertex $k$ of type 1; in particular, this requires that the lineages of these vertices have coalesced before time $t$ in the past.  

We next consider the initial state sampled from probability distribution $\mathbf{u}$ (defined in Section \ref{sec:fixweak}), which corresponds to setting a randomly chosen vertex to type 1 (with uniform probability) and all others vertices to type 0. Applying Eq.~\eqref{eq:coalescek}, we have
\begin{align}
\nonumber
\E^\circ_{\mathbf{u}} \big[S_i(t) \, S_j(t) \big] 
& = \frac{1}{N} \sum_{k \in G} \Prob^\CRW_{(i,j)}  \left[T_\coal < t,  X(t) = k \right]\\
& = \frac{1}{N} \Prob^\CRW_{(i,j)} \left[T_\coal < t \right].
\end{align}
The second equality follows from the law of total probability, since we have summed over all possible values of $X(t)$.  

We now move from a particular time $t$ to a time-average, using the operator $\langle \; \rangle^\circ_{\mathbf{u}}$. We calculate:
\begin{align}
\nonumber
\left \langle \frac{1}{N} - s_is_j \right \rangle^\circ_{\mathbf{u}} 
& = \int_0^\infty \left( \frac{1}{N} -\E^\circ_{\mathbf{u}} \big[S_i(t) \, S_j(t) \big] \right)   dt\\
\nonumber
& = \frac{1}{N} \int_0^\infty \left( 1-\Prob^\CRW_{(i,j)} \left[T_\coal< t \right] \right)  dt\\
\nonumber
& = \frac{1}{N} \E^\CRW_{(i,j)} \left[T_\coal \right] \\
\label{eq:sisjtau}
& = \frac{\tau_{ij}}{2N}.
\end{align}
(The factor of 2 arises from the translation from continuous time to discrete time.) In particular, this entails that for $n_1, n_2 \geq 0$
\begin{equation}
\label{eq:sinmin}
\left \langle \sum_{i \in G} \pi_i s_i \left( s_i^{(n_1)} - s_i^{(n_2)} \right) \right \rangle^\circ_{\mathbf{u}}  
= \frac{\tau^{(n_2)}-\tau^{(n_1)}}{2N}.
\end{equation}

\subsection{Recurrence relations for coalescence times}
\label{sec:recur}

The coalescence times $\tau_{ij}$ satisfy the recurrence relation
\begin{equation}
\label{eq:mrecur}
\tau_{ij} = \begin{cases} 0 & i=j\\
1 + \frac{1}{2}  \sum_{k \in G} \left(p_{ik} \tau_{jk} + p_{jk} \tau_{ik} \right)& i \neq j.
\end{cases}
\end{equation}
Eq.~\eqref{eq:mrecur} is a system of $\binom{N}{2}$ linear equations.  The connectedness of $G$ implies that this system has a unique solution.  All coalescence times $\tau_{ij}$ can be therefore obtained in polynomial time (see Section \ref{sec:computational} for a discussion of algorithms and their efficiency).

We will also make use of the concept of \emph{remeeting times}.  To define these we introduce a process called the \emph{remeeting random walk (RRW)}, which consists of two random walks $(X(t), Y(t))$ both starting at the same vertex $i \in G$: $X(0)=Y(0)=i$.  These random walks step independently until their time $T_\mathrm{remeet}$ of first remeeting, after which they step together.  As in the CRW, we can consider this process in either continuous time (with each walk proceeding at rate 1) or discrete time (with one of the two walkers stepping at each timestep).  We let $\tau_{ii}^+ = \tilde{\E}^{\mathrm{RRW}}_i[T_\mathrm{remeet}]$ denote the expected remeeting time in the discrete-time RRW.  ($\tau_{ii}^+$ is shortened to $\tau_i$ in the main text.) A simple recurrence argument shows that 
\begin{equation}
\label{eq:miirecur}
\tau_{ii}^+ = 1 + \sum_{j\in G} p_{ij} \tau_{ij}
\end{equation}

To derive a recurrence relation for the $\tau^{(n)}$, we expand according to Eq.~\eqref{eq:mrecur}:
\begin{align*}
\tau^{(n)} & = \sum_{i,j \in G} \pi_i p_{ij}^{(n)} \tau_{ij}\\
& =  \sum_{\substack{i,j \in G \\ i \neq j}} \pi_i p_{ij}^{(n)} \left( 1 + \frac{1}{2} \sum_{k \in G} (p_{ik} \tau_{jk} + p_{jk} \tau_{ik}) \right)\\
& = \sum_{i,j \in G} \pi_i p_{ij}^{(n)} \left( 1 + \frac{1}{2}  \sum_{k \in G} (p_{ik} \tau_{jk} + p_{jk} \tau_{ik}) \right) - \sum_{i \in G} \pi_i p_{ii}^{(n)}\left( 1+ \sum_{k \in G} p_{ik} \tau_{ik} \right)\\
& =  \sum_{i,j \in G} \pi_i p_{ij}^{(n)}  + \frac{1}{2}  \sum_{i,j,k \in G} \pi_j p_{ji}^{(n)} p_{ik} \tau_{jk} + \frac{1}{2} \sum_{i,j,k \in G} \pi_i p_{ij}^{(n)} p_{jk} \tau_{ik} - \sum_{i \in G} \pi_i p_{ii}^{(n)} \tau_{ii}^+\\
& =  \sum_{i \in G} \pi_i   + \frac{1}{2}  \sum_{j,k \in G} \pi_j p_{jk}^{(n+1)} \tau_{jk} + \frac{1}{2} \sum_{i,k \in G} \pi_i p_{ik}^{(n+1)} \tau_{ik} - \sum_{i \in G} \pi_i p_{ii}^{(n)} \tau_{ii}^+\\
& = 1 +  \tau^{(n+1)} - \sum_{i \in G} \pi_i p_{ii}^{(n)} \tau_{ii}^+.
\end{align*}

We therefore have the recurrence relation
\begin{equation}
\label{eq:minrecur}
\tau^{(n+1)} = \tau^{(n)} + \sum_{i \in G}\pi_i p_{ii}^{(n)} \tau_{ii}^+ - 1.
\end{equation}
In particular, we have
\begin{align}
\label{eq:m0}
\tau^{(0)} & = 0\\
\label{eq:m1}
\tau^{(1)} & = \sum_{i \in G}\pi_i \tau_{ii}^+ - 1\\
\label{eq:m2}
\tau^{(2)} & = \sum_{i \in G}\pi_i \tau_{ii}^+ \left( 1+p_{ii}^{(1)} \right) - 2\\
\label{eq:m3}
\tau^{(3)} & = \sum_{i \in G}\pi_i \tau_{ii}^+ \left(1+p_{ii}^{(1)} +p_{ii}^{(2)} \right) - 3.
\end{align}
We note that if $G$ has no self-loops then $p_{ii}^{(1)}=0$ for all $i \in G$.

Letting $n \to \infty$ in Eq.~\eqref{eq:minrecur}, (using Ces\`{a}ro means to get convergence in the case that random walks on $G$ are periodic) we obtain the identity
\begin{equation}
\label{eq:sumpi2mi}
\sum_{i \in G} \pi_i^2 \tau_{ii}^+ = 1.
\end{equation}
Eq.~\eqref{eq:sumpi2mi} can also be obtained by applying Kac's formula \cite[Corollary 2.24]{aldous2002reversible} to the random walk on $G \times G$.

\subsection{Time spent in different strategies}
\label{sec:timedifferent}

The probability that $i$ and $j$ employ different strategies at time $t$, under the neutral DB process with the uniform distribution $\mathbf{u}$ over initial placements of the mutant, can be calculated as
\begin{align*}
\Prob^\circ_\mathbf{u}[S_i(t) \neq S_j(t)] & = \frac{1}{N} \sum_{k \in G} \Prob^\CRW_{(i,j)} [T_\coal > t \wedge (X(t) = k \vee Y(t) = k)]\\
& = \frac{1}{N} \sum_{k \in G}  \Prob^\CRW_{(i,j)}[T_\coal > t, X(t) = k] \\
& \qquad + \frac{1}{N} \sum_{k \in G}  \Prob^\CRW_{(i,j)}[T_\coal > t, Y(t) = k]\\
& = \frac{2}{N} \Prob^\CRW_{(i,j)}[T_\coal > t].
\end{align*}

The total expected time that $i$ and $j$ spend as different strategies can be calculated as
\begin{align*}
\int_0^\infty \Prob^\circ_\mathbf{u}[S_i(t) \neq S_j(t)] \; dt & = \frac{2}{N} \int_0^\infty \Prob^\CRW_{(i,j)}[T_\coal > t] \; dt\\
& = \frac{2}{N} \E^\CRW_{(i,j)}[T_\coal]\\
& = \frac{\tau_{ij}}{N}.
\end{align*}

To connect to results described in the main text, we note that one unit of time in the continuous-time DB process corresponds, on expectation, to $N$ time-steps of the discrete-time DB process described in the main text.  Thus $\tau_{ij}$ is equal to the time that $i$ and $j$ employ different strategies in the neutral discrete-time DB process.  Let $\tau_\mathrm{abs}$ denote the expected time until absorption (mutant fixation or extinction) in the neutral discrete-time DB process. ($\tau_\mathrm{abs}$ is denoted $T$ in the main text.)  Then $\tau_\mathrm{abs} - \tau_{ij}$ is the expected time that $i$ and $j$ employ the same strategy, from mutant appearance until absorption, in the neutral discrete-time DB process.

\section{Conditions for success}
\label{sec:conditions}

\subsection{General case}
\label{sec:generalconditions}

Combining Eqs.~\eqref{eq:sfix} and \eqref{eq:sinmin}, we obtain the fixation probability of cooperators under weak selection:
\begin{equation}
\label{eq:fixtau}
\rho_\C = \frac{1}{N} +  
\frac{\delta}{2N}  \Big ( -c \tau^{(2)}
+ b \left( \tau^{(3)} - \tau^{(1)} \right) \Big)
+ \mathcal{O}(\delta^2).
\end{equation}
Cooperation is favored under weak selection, in the sense that $\rho_\C>1/N$ to first order in $\delta$, if and only if
\begin{equation}
\label{eq:successm}
 -c \tau^{(2)} + b \left( \tau^{(3)} - \tau^{(1)} \right) > 0.
\end{equation}
This is equivalent to Condition (1) and Eq.~(2) of the main text.  We show in Section \ref{sec:defector} that this is also the condition for $\rho_\D<1/N$.

Using Eqs.~\eqref{eq:m0}--\eqref{eq:m3}, we can rewrite Condition \eqref{eq:successm} as
\begin{equation}
\label{eq:DBcond}
-c \left(\sum_{i \in G} \pi_i \tau_{ii}^+ \left(1 + p_{ii}^{(1)} \right) - 2 \right) 
+ b \left( \sum_{i \in G} \pi_i \tau_{ii}^+ \left(p_{ii}^{(1)} + p_{ii}^{(2)} \right)- 2 \right) > 0.
\end{equation}

\subsection{Asymptotic formula}
\label{sec:pbar}

Here we derive Eq.~(3) of the main text, $(b/c)^*=1/\bar{p}$, which is an asymptotic expression that applies in a particular limit. Consider a graph $G$ with no self-loops.  Then $p_{ii}^{(1)}=0$ and we shorten $p_{ii}^{(2)}$ to $p_i$.  Substituting in Condition \eqref{eq:DBcond}, we obtain a critical benefit-cost ratio of
\begin{equation}
\label{eq:altcritbc}
\left( \frac{b}{c} \right)^* = \frac{\sum_{i \in G} \pi_i \tau_{ii}^+ - 2}{\sum_{i \in G} \pi_i \tau_{ii}^+p_i- 2 }.
\end{equation}

From Eq.~\eqref{eq:sumpi2mi} we have
\[
1  = \sum_{i \in G} \pi_i^2 \tau_{ii}^+ = \sum_{i \in G} \left( \frac{\pi_i}{p_i}\right)\pi_i \tau_{ii}^+ p_i 
\leq \max_{i \in G} \left( \frac{\pi_i}{p_i}\right) \sum_{i \in G} \pi_i \tau_{ii}^+ p_i.
\]
This implies that
\begin{equation}
\label{eq:sumbound}
\sum_{i \in G} \pi_i \tau_{ii}^+ p_i \geq \left( \max_{i \in G} \left( \frac{\pi_i}{p_i}\right)  \right)^{-1}.
\end{equation}

Now consider an infinite sequence of graphs $\{G_N\}_{N=N_0}^\infty$ for which \linebreak $\lim_{N \to \infty} \max_{i \in G_N} \left(\pi_i/p_i\right) = 0$.  (This is the ``locality property" discussed in the main text.)  For such a family, Eq.~\eqref{eq:sumbound} implies that the sums in the numerator and denominator of Eq.~\eqref{eq:altcritbc} diverge to infinity, resulting in  
\begin{equation}
\label{eq:asymcritbc}
\lim_{N \to \infty} \left( b/c \right)^* = \lim_{N \to \infty} 1/\bar{p},
\end{equation}
where $\bar{p} = \left( \sum_{i \in G_N} \pi_i \tau_{ii}^+p_i\right)/\left(\sum_{i \in G_N} \pi_i \tau_{ii}^+\right)$ is a weighted average of the $p_i$ with weights $\pi_i \tau_{ii}^+$.

\subsection{Weighted regular graphs}

$G$ is a \emph{weighted regular graph} if each vertex has the same distribution of weights to its neighbors.  This means that the set $\{w_{ij}\}_{j \in G}$ of outgoing weights is the same for each $i \in G$.  In this case, $\pi_i = 1/N$ for each $i \in G$, and Eq.~\eqref{eq:sumpi2mi} implies that $\tau_{ii}^+ =N$ for each $i \in G$.
The fixation probability of cooperators \eqref{eq:fixtau} becomes
\begin{equation}
\label{eq:weightedregular}
\rho_C = \frac{1}{N} + \frac{\delta}{2N} \Big( -c \left(N + Np^{(1)} - 2 \right) 
+ b \left( Np^{(1)} + Np^{(2)}- 2 \right) \Big)+ \mathcal{O}(\delta^2).
\end{equation}
Above, we have dropped the subscripts of $p_{ii}^{(1)}$ and $p_{ii}^{(2)}$ because they are the same at each vertex.

In the case that $G$ has no self-loops, $p^{(1)}=0$ and $p^{(2)}=1/\kappa$, where $\kappa = \left( \sum_{j \in G} p_{ij}^2 \right)^{-1}$ is the Simpson degree of $G$ \cite{allen2013spatial,allen2014games}.   The fixation probability \eqref{eq:weightedregular} simplifies to
\[
\rho_C = \frac{1}{N} + \frac{\delta}{2N} \big( -c \left(N - 2 \right) 
+ b \left( N/\kappa- 2 \right) \big)+ \mathcal{O}(\delta^2),
\]
and the critical benefit-cost ratio is
\begin{equation}
\label{eq:bckappa}
\left( \frac{b}{c} \right)^* = \frac{N - 2}{N/\kappa - 2 }.
\end{equation}
Eq.~\eqref{eq:bckappa} generalizes previous results for graphs with transitive symmetry \cite{Taylor,allen2014games} and for unweighted regular graphs \cite{chen2013sharp}.  In the unweighted case, the Simpson degree $\kappa$ is equal to the topological degree $k$.

\section{Starting from a single defector}
\label{sec:defector}

Suppose now that we start from a single defector (again placed uniformly at random). The game matrix can be written from the defector's point of view as
\begin{equation}
\label{eq:DPD}
\bordermatrix{
& \D & \C\cr
\D & 0 & b\cr
\C & -c & b-c }.
\end{equation}

From Eq.~\eqref{eq:D's} we see that $D'(\vs)$---the first-order term of the expected instantaneous rate of degree-weighted frequency change---is unaffected by the addition of a constant to the payoff matrix.  Thus we can, without affecting weak-selection fixation probability, add $-b+c$ to each entry of the matrix \eqref{eq:DPD}, resulting in the payoff matrix
\[
\bordermatrix{
& \D & \C\cr
\D & -b+c & c \cr
\C & -b & 0 }.
\]
This is the original payoff matrix \eqref{eq:PD} with the roles of C and D reversed and with $b$ and $c$ replaced by their opposites.  The fixation probability $\rho_\D$ is therefore obtained by replacing $b$ and $c$ by their opposites in Eq.~\eqref{eq:fixtau}:
\[
\rho_\D = \frac{1}{N} +  
\frac{\delta}{2N}  \Big ( c \tau^{(2)}
- b \left( \tau^{(3)} - \tau^{(1)} \right) \Big)
+ \mathcal{O}(\delta^2).
\]
In particular, this shows that Condition \eqref{eq:successm} implies $\rho_\C>1/N>\rho_\D$, and the opposite inequality in Condition \eqref{eq:successm} implies $\rho_\C<1/N<\rho_\D$.

\section{Arbitrary $2 \times 2$ games}
\label{sec:sigma}

We now consider arbitrary $2 \times 2$ matrix games of the form \eqref{eq:game}.  The general condition for success follows from the Structure Coefficient Theorem \cite{Corina}, which states that the condition for A to be favored over B (in the sense $\rho_\mathrm{A} > \rho_\mathrm{B}$ under weak selection) takes the form $\sigma a + b > c + \sigma d$ for some structure coefficient $\sigma$ that is independent of the game. Because of this independence, we can obtain $\sigma$ from our analysis of the simplified Prisoner's Dilemma \eqref{eq:PD}:
\begin{equation}
\sigma = \frac{(b/c)^*+1}{(b/c)^*-1}.
\end{equation}
Combining with Condition \eqref{eq:successm}, we obtain
\begin{equation}
\label{eq:sigma}
\sigma = \frac{-\tau^{(1)} + \tau^{(2)} + \tau^{(3)}}{\tau^{(1)} + \tau^{(2)} - \tau^{(3)}},
\end{equation}
as stated in the main text.

We note that the condition $\sigma a + b > c + \sigma d$ tells us that $\rho_\mathrm{A} > \rho_\mathrm{B}$, but does not, in general, tell us how these fixation probabilities compare with the neutral fixation probability $1/N$  \cite{TaylorFixProb}.  However, for games that satisfy $a+d=b+c$ (a property known as ``equal gains from switching") $\sigma a + b > c + \sigma d$ implies $\rho_\mathrm{A} > 1/N > \rho_\mathrm{B}$.

\section{Bounds on $(b/c)^*$ and $\sigma$}
\label{sec:sigmapos}

Here we prove the bounds $|(b/c)^*|>1$ and $\sigma>0$.  These will follow from the inequalities $\tau^{(1)} + \tau^{(2)} \geq \tau^{(3)}$ and $\tau^{(2)} + \tau^{(3)} \geq \tau^{(1)}$.   We begin with a lemma:

\begin{lemma}
\label{lem:piineq}
Let $G$ be a connected weighted graph (undirected, possibly with self-loops).  
For each vertex $i$ of  $G$,
\begin{equation}
\label{eq:piineq}
1 + p_{ii}^{(2)} \geq 4\pi_i.
\end{equation}
Equality occurs if and only if all edges not adjcacent to $i$ have weight zero.
\end{lemma}

\begin{proof}  Let vertex $i \in G$ be arbitrary and fixed. Our proof uses a variational approach: we hold constant the weights of the edges $w_{ij}$ adjacent to $i$, but allow the weights of edges not adjacent to $i$ to vary.  We will also allow the edges not adjacent to $i$ to be directed---that is, $w_{jk}$ may differ from $w_{kj}$ for $j,k\neq i$---while maintaining that the edges adjacent to $i$ be undirected ($w_{ij}=w_{ji}$ for all $j$).

In this generalized setting, we denote in-degrees and out-degrees as follows:
\[
w_j^\mathrm{in} = \sum_{k \in G} w_{kj}, \qquad
w_j^\mathrm{out} = \sum_{k \in G} w_{jk}.
\]
The total edge weight of $G$ is again denoted $W$:
\begin{equation}
\label{eq:Wsum}
W = \sum_{j,k \in G} w_{jk} = \sum_{j \in G} w_j^\mathrm{in} = \sum_{j \in G} w_j^\mathrm{out}.
\end{equation}
For vertex $i$ we have $w_i^\mathrm{in} = w_i^\mathrm{out}=w_i$.  The $\pi_i$ in \eqref{eq:piineq} is to be understood as $w_i/W$. It is useful to observe the identity
\begin{equation}
\label{eq:Wsum2}
W=2w_i + 2 \sum_{h,k \neq i} w_{hk}.
\end{equation}

We define the function
\begin{equation}
\label{eq:Fdef}
F\left( (w_{h k})_{\substack{h,k \neq i\\h \neq k}} \right) = 
w_i + \sum_{j \in G} \frac{w_{ij}^2}{w_j^\mathrm{out}} - 4 \frac{w_i^2}{W}.
\end{equation}
Note that if all edges not adjacent to $i$ have weight zero (i.e.~$w_{hk}=0$ for all $h, k \neq i$), then $w_j^\mathrm{out}=w_{ji}=w_{ij}$ for each $j \neq i$ and $w_i=W/2$, whereupon substituting in \eqref{eq:Fdef} gives $F=0$.  We will prove that this is the minimum value of $F$.  Dividing by $w_i$ will then yield
\[
1+\sum_{j \in G} \frac{w_{ij}^2}{w_i w_j^\mathrm{out}} - 4 \frac{w_i}{W} \geq 0,
\]
which is equivalent to the desired inequality \eqref{eq:piineq}.

To prove that zero is the minimum value of $F$, we first observe from \eqref{eq:Wsum2} that if $\sum_{h, k \neq i} w_{h k} \geq w_i$ then $w_i \leq W/4$.  In this case, writing  the first and third terms of $F$ as $w_i(1-4w_i/W)$, we see that $F>0$.  Therefore, to minimize $F$, it suffices to restrict the domain of $F$ to the simplex 
\[
\Delta = \left\{(w_{h k})_{h, k \neq i} \; \Bigg| \; w_{h k} \geq 0, \; \sum_{h, k \neq i} w_{h k} < w_i \right\}.
\]

We now look for critical points of $F$ on the interior of $\Delta$, as well as critical points of $F$ restricted to the one of the faces of $\Delta$ where one or more of the $w_{h k}$ are zero.  Such critical points of $F$ have the property that for all $h, k \neq i$, either $w_{h k} = 0$ or $\frac{\partial F}{\partial w_{h k}} = 0$.
In the latter case, we have
\begin{equation}
\label{eq:partialF}
0 = \frac{\partial F}{\partial w_{h k}}  =  - \frac{w_{ih}^2}{\left( w_h^\mathrm{out} \right)^2} 
 + 4 \frac{w_i^2}{W^2},
\end{equation}
which implies that 
\begin{equation}
\label{eq:critpts}
\frac{w_{ih}}{w_h^\mathrm{out}} = 2\frac{w_i}{W}.
\end{equation}
We note in particular that Eq.~\eqref{eq:critpts} is independent of $k$.  This implies a stronger characterization of potential minimizing points of $F$: They must satisfy the property that for all $h \neq i$, either Eq.~\eqref{eq:critpts} holds or else $w_{h k} = 0$ for all $k \neq i$.  Equivalently, for potential minimizing points of $F$, the vertices $h$ other than $i$ can be partitioned into two disjoint subsets, labeled $V_1$ and $V_2$, such that $w_h^\mathrm{out} = w_{ih}$ for $h \in V_1$ and $w_h^\mathrm{out} =  (W/2w_i)\,  w_{ih}$ for $h \in V_2$.

We now apply Eq.~\eqref{eq:Wsum}:
\begin{align*}
W & = w_i + \sum_{h \in V_1} w_h^\mathrm{out} + \sum_{h \in V_2} w_h^\mathrm{out}\\
& = w_i + \sum_{h \in V_1} w_{ih} + \frac{W}{2w_i} \sum_{h \in V_2} w_{ih}\\
& = 2w_i + \left(\frac{W}{2w_i} - 1 \right) \sum_{h \in V_2} w_{ih}.
\end{align*}
Rearranging algebraically, we obtain
\begin{equation*}
\left( 2w_i - \sum_{h \in V_2} w_{ih} \right) (W-2w_i) = 0.
\end{equation*}
The first factor on the left-hand side is necessarily positive; therefore all potential minimizing points of $F$ satisfy $w_i = W/2$. In this case, all edges not adjacent to $i$ have weight zero, which we have already shown implies that $F=0$.  Thus $F \geq 0$ with equality if and only if all edges not adjacent to $i$ have weight zero. Dividing Eq.~\eqref{eq:Fdef} by $w_i$ completes the proof.
\end{proof}

Positivity of the numerator and the denominator of $\sigma$ in Eq.~\eqref{eq:sigma} for $N \geq 3$ follows from the above lemma:

\begin{theorem}
\label{thm:tauineq}
For any connected weighted graph $G$, the meeting times $\tau^{(1)}$, $\tau^{(2)}$, and $\tau^{(3)}$ satisfy
\renewcommand{\labelenumi}{(\alph{enumi})}
\begin{enumerate}
\item $\tau^{(1)} + \tau^{(2)} \geq \tau^{(3)}$,
\item $\tau^{(2)} + \tau^{(3)} \geq \tau^{(1)}$,
\item $\tau^{(1)} + \tau^{(3)} \geq \tau^{(2)} + 2$.
\end{enumerate}
In (a) and (b), equality occurs if and only if $G$ has two vertices and no self-loops.  In (c), equality occurs if and only if $G$ has two vertices.
\end{theorem}

\begin{proof}
For (a) we use Eq.~\eqref{eq:m1}--\eqref{eq:m3}:
\begin{align*}
\tau^{(3)} - \tau^{(1)}  & = \sum_{i \in G}\pi_i \tau_{ii}^+ \left( p_{ii}^{(1)} +p_{ii}^{(2)} \right) - 2\\
& \leq \sum_{i \in G}\pi_i \tau_{ii}^+ \left( p_{ii}^{(1)} +1 \right) - 2\\
& = \tau^{(2)}.
\end{align*}
Equality occurs if and only if $p_{ii}^{(2)}=1$ for all $i$, which is only possible when $G$ has two vertices and no self-loops.

For (b) and (c), we combine Eqs.~\eqref{eq:m3} and \eqref{eq:sumpi2mi} with Lemma \ref{lem:piineq}:
\begin{align}
\nonumber
\tau^{(3)} & = \sum_{i \in G}\pi_i \tau_{ii}^+\left(1+ p_{ii}^{(1)} + p_{ii}^{(2)} \right) - 3 \\
\nonumber
& \geq  \sum_{i \in G}\pi_i \tau_{ii}^+ \left( p_{ii}^{(1)} + 4 \pi_i\right)- 3\\ 
\label{eq:t3ineq}
& = \sum_{i \in G}\pi_i \tau_{ii}^+ p_{ii}^{(1)} +1.
\end{align}
Equality in \eqref{eq:t3ineq} attains if and only if when $1+p_{ii}^{(2)}=4\pi_i$ for all vertices $i$ of $G$.  By Lemma \ref{lem:piineq}, this means that each vertex $i$ has the property that all edges not adjacent to $i$ have weight 0.  The only connected graphs with this property are those of size two.

Inequality (c) is now obtained by using Eqs.~\eqref{eq:m1} and \eqref{eq:m2} to rewrite the last line of \eqref{eq:t3ineq} as $\tau^{(2)}-\tau^{(1)}+2$.  For (b), we rewrite the last line of Eq.~\eqref{eq:t3ineq} as $\tau^{(1)}-\tau^{(2)}+2 \sum_{i \in G}\pi_i \tau_{ii}^+ p_{ii}^{(1)}$, which is greater than or equal to $\tau^{(1)}-\tau^{(2)}$ with equality if and only if $G$ has no self-loops.
\end{proof}

We now turn to the bounds on $(b/c)^*$ and $\sigma$.  

\begin{corollary}
For any graph of size $N \geq 3$, we have the bounds $\sigma > 0$ and $|(b/c)^*|>1$.  For graphs of size $N=2$, $\rho_A = \rho_B = 1/2$ regardless of the game \eqref{eq:game}; thus both $(b/c)^*$ and $\sigma$ are undefined.
\end{corollary}

\begin{proof}
We begin with the case $N \geq 3$, for which the inequalities in Theorem \ref{thm:tauineq} are strict.  Inequalities (b) and (a) assert that the numerator and denominator, respectively, of Eq.~\eqref{eq:sigma} for $\sigma$ are positive, and therefore $\sigma > 0$.  The bound $|(b/c)^*|>1$ can be obtained by writing
\[
\left( \frac{b}{c} \right)^* = \frac{\sigma + 1}{\sigma - 1},
\]
and noting that the right-hand side has absolute value greater than 1 for $\sigma >0$.

For the case $N=2$, we observe that, from an initial state with one A and one B, fixation is determined by the first death event.  Whichever type is chosen for death, the other type becomes fixed, regardless of the game.  Thus $(b/c)^*$ and $\sigma$ are undefined.
\end{proof}

Note that $(b/c)*$ can come arbitrarily close to 1 (equivalently, $\sigma$ can be arbitrarily large), as in Fig.~\ref{fig:promote}c of the main text.  It is not immediately clear whether $(b/c)*$ can come arbitrarily close to -1 (equivalently, whether $\sigma$ can be arbitrarily close to 0).

\section{Arbitrary mutation rates}
\label{sec:mutation}

Here we relax the assumption that offspring always inherit the type of the parent.  We introduce mutation of arbitrary probability $0 \leq u \leq 1$ per reproduction.  With probability $1-u$, a new offspring inherits the type of the parent.  Otherwise a mutation occurs, which is equally likely to result in either type.  So, for example, the offspring of a type A individual has probability $1-u$ of inheriting type A, probability $u/2$ of mutating but remaining type A, and probability $u/2$ of mutating to type B.

For $u>0$ there is a unique stationary distribution over states of the evolutionary process, which we call the \emph{mutation-selection stationary distribution} or \emph{MSS distribution} \cite{allen2014measures,allen2014games}.  We use $\E_\mathrm{MSS}$ to denote the expectation of a quantity under the MSS distribution, and $\E_\mathrm{MSS}^\circ$ to denote the same expectation under neutral drift ($\delta=0$).  

Following Tarnita and Taylor \cite{tarnita2014measures}, we say that type A is favored by selection if $\E_\mathrm{MSS}[\hat{x}]>\frac{1}{2}$; that is, if A has greater degree-weighted abundance than B in the MSS distribution.  Eq.~(B.3) of Tarnita and Taylor \cite{tarnita2014measures} implies that for weak selection and DB updating, A is favored if and only if 
\begin{equation}
\label{eq:tarnitataylor}
\E_\mathrm{MSS}^\circ[D']>0.
\end{equation}

Evaluating $\E_\mathrm{MSS}^\circ[D']$ requires analyzing the assortment of types under the MSS distribution.  This can be accomplished using the method of identity-by-descent \cite{malecot1948IBD,RoussetBook,Taylor}.  Two individuals are \emph{identical by descent} (IBD) if no mutation separates them from their common ancestor.    We let $q_{ij}$ denote the stationary probability that the occupants of $i$ and $j$ are IBD to each other.  These IBD probabilities can be rigorously defined using the notion of the IBD-enriched Markov chain \cite{allen2014games}. IBD probabilities on an arbitrary weighted, connected graph $G$ can be obtained as the unique solution to the system of equations \cite{allen2014games}:
\begin{equation}
\label{eq:qsystem}
q_{ij} = \begin{cases} 
1 & i=j\\
\frac{1-u}{2} \sum_{k \in G} (p_{ik} q_{jk} + p_{jk} q_{ik}) & i \neq j. 
\end{cases}
\end{equation}
By generalizing Lemma 3 of Allen \& Nowak \cite{allen2014games}, one can use IBD probabilities to calculate assortment under the neutral MSS distribution:
\begin{equation}
\label{eq:sisjq}
\E^\circ_\mathrm{MSS} \left[ s_i s_j  \right ]  =  \frac{1+q_{ij}}{4}.
\end{equation}

For the donation game \eqref{eq:PD}, combining Eqs.~\eqref{eq:D'sPD}, \eqref{eq:tarnitataylor}, and \eqref{eq:sisjq}, we obtain that cooperation is favored under weak selection if and only if 
\begin{equation}
\label{eq:qcondition}
-c \left( 1 - q^{(2)} \right) + b \left(q^{(1)} - q^{(3)} \right) > 0.
\end{equation}
Above, according to our convention, we have set $q^{(n)} = \sum_{i,j \in G} \pi_i p_{ij}^{(n)} q_{ij}$.  \

For an arbitrary game of the form \eqref{eq:game}, the Structure Coefficient Theorem \cite{Corina} implies that A is favored under weak selection if and only if $\sigma a + b > c + \sigma d$ where 
\[
\sigma = \frac{1+q^{(1)}-q^{(2)}-q^{(3)}}{1-q^{(1)}-q^{(2)} + q^{(3)}}.
\]
This result generalizes Theorem 15 of Allen \& Nowak \cite{allen2014games} to arbitrary weighted graphs.

To connect Condition \eqref{eq:qcondition} to our main result \eqref{eq:successm}, we apply an established connection between coalescence times and identity-by-descent probabilities \cite{slatkin1991inbreeding,van2015social}.  Since each step in the discrete-time coalescing random walk corresponds to a replacement, and mutations occur with probability $u$ per replacement, the IBD probability $q_{ij}$ has the low-mutation expansion
\begin{equation}
\label{eq:qtau}
q_{ij} = \tilde{\E}^\mathrm{CRW}_{(i,j)} \left [(1-u)^{T_\mathrm{coal}} \right] = 1 - u\tau_{ij} + \mathcal{O}(u^2) \qquad (u \to 0).
\end{equation}
Eq.~\eqref{eq:qtau} can also be obtained directly by comparing Eq.~\eqref{eq:qsystem} for $q_{ij}$ with Eq.~\eqref{eq:mrecur} for $\tau_{ij}$.  The equivalence of Conditions \eqref{eq:successm} and \eqref{eq:qcondition} in the low-mutation limit ($u \to 0$) follows directly from Eq.~\eqref{eq:qtau}.

\section{Variations on the model}
\label{sec:variations}

Here we consider three variations on the model: using accumulated rather than averaged payoffs, using different interaction and replacement graphs, and using Birth-Death rather than Death-Birth updating.  In each case we obtain the exact condition for success in terms of coalescence times.  Although we do not explicitly combine these variations (e.g.~Birth-Death updating with accumulated payoffs), such combinations can be analyzed using straightforward combinations of the methods described here.

\subsection{Accumulated payoffs}
\label{sec:accum}

Accumulated payoffs means that the payoffs to vertex $i$ from its neighbors are multiplied by the corresponding edge weights and summed, without normalization.  For the simplified Prisoners' Dilemma \eqref{eq:PD}, the accumulated payoff to vertex $i$ in state $\vs$ is given by
\begin{align*}
f_i(\vs) & = w_i \left(-c s_i + b s_i^{(1)} \right)\\
& = W\pi_i \left(-c s_i + b s_i^{(1)} \right).
\end{align*}

To derive $(b/c)^*$ for accumulated payoffs, we calculate $D'(\vs)$, the first-order term in the instantaneous rate of change in $\hat{s}$ from state $\vs$, starting with Eq.~\eqref{eq:D's}:
\begin{align*}
D'(\vs) & =  \sum_{i \in G} \pi_i  s_i \left( f_i^{(0)}(\vs) - f_i^{(2)}(\vs) \right)\\
& = W \sum_{i \in G} \pi_i  s_i \left[ \pi_i  \left(-c s_i + b s_i^{(1)} \right) - \sum_{j \in G} p_{ij}^{(2)} \pi_j \left(-c s_j + b s_j^{(1)} \right) \right] \\
& = W \left[ -c \left( \sum_{i \in G}\pi_i^2 s_i^2 -\sum_{i,j \in G} \pi_i p_{ij}^{(2)} \pi_j s_i s_j \right)
+ b \left( \sum_{i \in G}\pi_i^2 s_is_i^{(1)} -\sum_{i,j \in G} \pi_i p_{ij}^{(2)} \pi_j s_i s_j^{(1)} \right) \right]\\
& = W \left[ -c \left( \sum_{i \in G}\pi_i^2 s_i^2 -\sum_{i,j \in G} \pi_j^2 p_{ji}^{(2)} s_i s_j \right)
+ b \left( \sum_{i \in G}\pi_i^2 s_is_i^{(1)} -\sum_{i,j \in G} \pi_j^2 p_{ji}^{(2)}  s_i s_j^{(1)} \right) \right]\\
& = W \left[ -c  \sum_{i,j \in G}\pi_i^2 p_{ij}^{(2)} \left( s_i^2 - s_is_j \right)
+ b  \sum_{i,j \in G}\pi_i^2 p_{ij}^{(2)} \left( s_i s_i^{(1)} - s_js_i^{(1)} \right) \right]\\
& = W \left[ -c  \sum_{i,j \in G}\pi_i^2 p_{ij}^{(2)} \left( s_i^2 - s_is_j \right)
+ b  \sum_{i,j,k \in G}\pi_i^2 p_{ij}^{(2)} p_{ik} \left( s_is_k - s_js_k \right) \right].
\end{align*}

Now combining with Eqs.~\eqref{eq:rhoweaku1} and \eqref{eq:sisjtau}, we obtain the fixation probability of cooperation:
\[
\rho_C = \frac{1}{N} + \frac{\delta W}{2N} \left(
-c \sum_{i,j \in G}\pi_i^2 p_{ij}^{(2)} \tau_{ij}
+ b  \sum_{i,j,k \in G}\pi_i^2  p_{ij}^{(2)} p_{ik} \left( \tau_{jk} - \tau_{ik} \right) \right)
+ \mathcal{O}(\delta^2).
\]
The critical benefit-cost ratio is therefore
\[
\left( \frac{b}{c} \right)^* = \frac{ \sum_{i,j \in G}\pi_i^2 p_{ij}^{(2)} \tau_{ij}}{ \sum_{i,j,k \in G}\pi_i^2  p_{ij}^{(2)} p_{ik} \left( \tau_{jk} - \tau_{ik} \right)}.
\]

\subsection{Different interaction and replacement graphs}
\label{sec:interactionreplacement}

Here we consider a variation in which the edge weights for game interaction differ from those for replacement \cite{OhtsukiBreaking,ohtsuki2007breakingJTB,Taylor,allen2014games}. In this case, the population structure is represented by a pair of weighted graphs $(G,I)$ which have the same vertex set $V$.  We assume that the replacement graph $G$ is connected (so that the population is unitary), but the interaction graph $I$ need not be.  

We define an $(n,m)$-random walk to be a random walk consisting of $n$ steps according to the weights of $G$, followed by $m$ steps according to the weights of $I$.  Let $p^{(n,m)}_{ij}$ denote the probability that an $(n,m)$-random walk starting at $i$ terminates at $j$. The payoff $f_i(\vs)$ to vertex $i$ in state $\vs$ can be written analogously to Eq.~\eqref{eq:fPD} as
\[
f_i(\vs) = -cs_i + b \sum_{j \in V} p^{(0,1)}_{ij} s_j.
\]
(This expression uses averaged payoffs; the extension to accumulated payoffs is straightforward.) 

The results in Section \ref{sec:fixweak} carry over verbatim, with the understanding that $\pi_i$ is defined using the weights for the replacement graph $G$. Following the steps in Section \ref{sec:shortterm}, we obtain 
\begin{equation}
\label{eq:D'sGI}
D'(\vs) = \sum_{i \in V} \pi_i s_i \left( -c \left(s_i^{(0,0)} - s_i^{(2,0)} \right) + b \left(s_i^{(0,1)} - s_i^{(2,1)} \right) \right).
\end{equation}
Above, we have adopted the notation $s_i^{(n,m)} = \sum_{j \in V} p_{ij}^{(n,m)} s_j$. Applying Eq.~\eqref{eq:sisjtau}, we find an analogue of Eq.~\eqref{eq:sinmin}:
\begin{equation}
\label{eq:sinminGI}
\left \langle \sum_{i \in G} \pi_i s_i \left( s_i^{(n_1,m_1)} - s_i^{(n_2,m_2)} \right) \right \rangle^\circ_{\mathbf{u}}  
= \frac{\tau^{(n_2,m_2)}-\tau^{(n_1,m_1)}}{2N},
\end{equation}
where
\[
\tau^{(n,m)} = \sum_{i,j \in V} \pi_i p_{ij}^{(n,m)} \tau_{ij}.
\]
Substituting Eqs.~\eqref{eq:D'sGI} and \eqref{eq:sinminGI} into Eq.~\eqref{eq:sfix} we obtain the fixation probability of cooperation:
\[
\rho_\C = \frac{1}{N} +  
\frac{\delta}{2N} \Big(  -c \tau^{(2,0)} + b \left( \tau^{(2,1)} - \tau^{(0,1)} \right) \Big)
+ \mathcal{O}(\delta^2).
\]
The critical benefit-cost ratio is therefore
\[
\left(\frac{b}{c}\right)^* =  \frac{\tau^{(2,0)}}{ \tau^{(2,1)} - \tau^{(0,1)}}.
\]

\subsection{Birth-Death updating}

In Birth-Death (BD) updating \cite{Ohtsuki}, first an individual is chosen to reproduce, with probability proportional to its reproductive rate.  The offspring then replaces a random neighbor, chosen with probability proportional to edge weight.  

We study a continuous-time analogue of the birth-death process, in which site $i$ replaces site $j$ at rate $F_i(\vs) p_{ij}$:
\begin{equation}
\operatorname{Rate}[i \rightarrow j](\vs) =  F_i(\vs) p_{ij}.
\end{equation}

As for DB updating, we suppose that a new type is equally likely to arise at each vertex.  This assumption is mathematically convenient, although it is arguably more natural to suppose that mutations arise in proportion to how often a vertex is replaced \cite{allen2014measures,tarnita2014measures,allen2015molecular}.  

Our methods for DB updating largely carry over to BD, with some modifications that we describe here.

\subsubsection{Fixation probability under weak selection}

Instead of weighting each vertex $i$ by its degree, we weight by the inverse degree $1/w_i$.  The inverse degree can be understood as the reproductive value of vertex $i$ under BD updating \cite{maciejewski2014reproductive}, and is proportional to the fixation probability of a neutral mutation arising at $i$ \cite{maciejewski2014reproductive,allen2015molecular}.  We therefore replace $\hat{s}$ with the quantity
\[ 
\tilde{s} =  \frac{1}{\tilde{W}}\sum_{i \in G} \frac{s_i}{w_i},  \qquad \tilde{W} = \sum_{i \in G} \frac{1}{w_i}.
\]
The arguments of Section \ref{sec:fixweak} carry over using $\tilde{s}$ in place of $\hat{s}$, leading to an analogue of Eq.~\eqref{eq:rhoweak2}:
\begin{equation*}
\rho_{\vs_0}= \tilde{s}_0 + \delta \langle D' \rangle^\circ_{\vs_0} + \mathcal{O}(\delta^2).
\end{equation*}
Above, $D'(\vs)$ is the first-order term of the expected instantaneous rate of change in $\tilde{s}$, defined by the following analogues of Eqs.~\eqref{eq:Ddef} and \eqref{eq:Dweak}:
\begin{align*}
\E \left[\tilde{S}(t+\epsilon)-\tilde{S}(t) \big| S(t)=\vs \right] \; & = \;  D(\vs) \epsilon + o(\epsilon) \qquad (\epsilon \to 0^+)\\
D(\vs) \; & = \; \delta D'(\vs) + \mathcal{O}(\delta^2) \quad (\delta \to 0).
\end{align*}
When we consider the uniform distribution $\mathbf{u}$ for initial mutant appearance, Eq.~\eqref{eq:rhoweaku1} for $\rho_\mathbf{u}$ carries over verbatim.  

We calculate the instantaneous rate of change $D(\vs)$ as follows:
\begin{align*}
D(\vs) &= \frac{1}{\tilde{W}} \sum_{i,j\in G} \operatorname{Rate}[i \rightarrow j] \frac{s_i - s_j}{w_j}\\
&= \frac{1}{\tilde{W}}\sum_{i,j \in G} \frac{w_{ij}}{w_i w_j}  F_i(\vs)(s_i - s_j) \\
&= \frac{\delta}{\tilde{W}} \sum_{i,j \in G} \frac{w_{ij}}{w_i w_j} f_i(\vs) (s_i - s_j) + \mathcal{O}(\delta^2)\\
&=\frac{\delta}{\tilde{W}} \sum_{i,j \in G} \frac{w_{ij}}{w_i w_j}  s_i (f_i(\vs) - f_j(\vs)) + \mathcal{O}(\delta^2)\\
&=\frac{\delta}{\tilde{W}} \sum_{i,j \in G} \frac{w_{ij}}{w_i w_j} \left( -c ( s_i ^2 - s_i s_j)  + b \sum_{k \in G} \left(  p_{ik} s_i s_k - p_{jk} s_i s_k \right)\right) + \mathcal{O}(\delta^2).
\end{align*}
Hence we get, after swapping $i$ and $j$ second part of the second sum,
\begin{equation}
\label{eq:DweakBD}
D'(\vs)  = \frac{1}{\tilde{W}} \left( - c \sum_{i,j \in G} \frac{w_{ij}}{w_i w_j} (s_i^2 - s_i s_j) + b \sum_{i,j,k \in G} \frac{w_{ij}}{w_i w_j}p_{ik} (s_i s_k - s_j s_k) \right)
\end{equation}

\subsubsection{Coalescence and assortment}

The final key step is finding the appropriate modification of the coalescing random walk (CRW) used in Section~\ref{sec:CRW}. The rate of stepping from $i$ to $j$ must correspond to the rate at which the $j \to i$ replacement occurs at neutrality, which is $w_{ij}/w_j$ for BD updating. Therefore the CRW for BD is a continuous-time process in which steps from $i$ to $j$ occur at rate $w_{ij}/w_j$ (instead of at the usual rate $p_{ij}=w_{ij}/w_i$).  

We let $\tilde{\tau}_{ij}$ denote the coalescence time from vertices $i$ and $j$ under this modified CRW.  (Note that the  $\tilde{\tau}_{ij}$ are defined in continuous time) These coalescence times satisfy the system of equations
\begin{equation}
\label{eq:tauBD}
\tilde{\tau}_{ij} = \begin{cases}
0 & i=j\\
\displaystyle \frac{1}{\sum_{\ell \in G} \frac{w_{i\ell} + w_{j \ell}}{w_\ell} }
\left( 1 + \sum_{k \in G} \frac{w_{ik} \tilde{\tau}_{ik} + w_{jk}\tilde{\tau}_{jk}}{w_k} \right) & i \neq j.
\end{cases}
\end{equation}
Again, the connectedness of $G$ implies that this system has a unique solution, which can be obtained in polynomial time (see Section \ref{sec:computational}).

Armed with this modified CRW, the arguments of Section \ref{sec:assortment} can be readily adapted to BD updating.  In particular, Eq.~\eqref{eq:sisjtau} has the analogue
\begin{equation}
\label{eq:sisjtauBD}
\left \langle \frac{1}{N} - s_is_j \right \rangle^\circ_{\mathbf{u}} = \frac{\tilde{\tau}_{ij}}{N}
\end{equation}
In contrast to Eq.~\eqref{eq:sisjtau}, there is no factor of 2 in the denominator on the right-hand side of Eq.~\eqref{eq:sisjtauBD}.  This is because the $\tilde{\tau}_{ij}$ are defined in continuous time, whereas the $\tau_{ij}$ are defined in discrete time.

\subsubsection{Condition for success}

Combining Eqs.~\eqref{eq:rhoweaku1}, \eqref{eq:DweakBD} and \eqref{eq:sisjtauBD}, we obtain the fixation probability of cooperation
\[
\rho_C = \frac{1}{N} + \frac{\delta}{\tilde{W}N} \left( 
- c \sum_{i,j \in G} \frac{w_{ij}}{w_i w_j} \tilde{\tau}_{ij} 
+ b \sum_{i,j,k \in G} \frac{w_{ij}}{w_i w_j}p_{ik} (\tilde{\tau}_{jk} - \tilde{\tau}_{ik} ) \right)
+ \mathcal{O}(\delta^2).
\]
The critical benefit-cost ratio is therefore
\[ 
\left(\frac{b}{c}\right)^* = \frac{ \sum_{i,j \in G} \frac{w_{ij}}{w_i w_j} \tilde{\tau}_{ij}}{ \sum_{i,j,k \in G} \frac{w_{ij}}{w_i w_j}p_{ik} (\tilde{\tau}_{jk} - \tilde{\tau}_{ik})}. 
\]

\section{Examples}
\label{sec:examples}

\subsection{Island model}
\label{sec:island}

Consider a population subdivided into $n$ islands of size $N_1, \ldots, N_n$.  Weights are 1 for distinct vertices on the same island and $m < 1$ for vertices on different islands.  Therefore, the weighted degree of a vertex on island $i$ is $w_i=N_i-1 + m(N-N_i)$.

We let $\tau_{ii'}$ denote the coalescence time for two distinct vertices on island $i$, and $\tau_{ij}$ the coalescence time for a pair on islands $i$ and $j$.  The recurrence relations \eqref{eq:mrecur} become

\begin{equation}
\label{eq:tauisland}
\begin{split}
\tau_{ii'} & = 1 + \frac{1}{N_i-1 + m(N-N_i)} \left( (N_i-2) \tau_{ii'} + m \sum_{j \neq i} N_j \tau_{ij} \right)\\[4mm]
\tau_{ij} & = 1 + \frac{(N_i-1) \tau_{ij} + m \sum_{k \neq i,j} N_k \tau_{kj} + m(N_j-1) \tau_{jj'}}{2(N_i-1 + m(N-N_i))} \\[2mm]
& \quad \qquad+ \frac{(N_j-1) \tau_{ij} + m \sum_{k \neq i,j} N_k \tau_{ik} + m(N_i-1) \tau_{ii'}}{2(N_j-1 + m(N-N_j))}.
\end{split}
\end{equation}

\subsubsection{Two islands}

In the case of two islands, we obtain the critical $b/c$ ratio exactly, by solving Eq.~\eqref{eq:tauisland} with the aid of Mathematica and applying Eq.~\eqref{eq:altcritbc}.  We obtain an answer in the form $(b/c)^* = \mathrm{num}/\mathrm{denom}$, where the numerator and denominator are given respectively by
\begin{multline*}
\mathrm{num} = \left((N m+N-2)^2-D^2 (m-1)^2\right)^2 \\
\times \bigg \{ D^6 (m-1)^3 m-D^4 (m-1)^2 \left(3 N^2 m (m+1)-2 N \left(m^2+5
   m-2\right)+10 m-6\right)\\
   +D^2 N (m-1) \Big[3 N^3 m (m+1)^2-2 N^2 \left(2 m^3+9 m^2+6 m-1\right)\\
   +2 N \left(7 m^2+6 m-7\right)-8 (m-2)\Big] \\
   -(N-2) N^2 (m+1) (N m+2) (N m+N-2)^2\bigg \},
\end{multline*}
\begin{multline*}
\mathrm{denom} = 4 \bigg\{- D^8 (m-1)^4 m \left(N (m-1)+m^2-3 m+4\right)\\
-D^6 (m-1)^3 \Big[N^3 m \left(m^3-m^2-3
   m-1\right)  +N^2 \left(-5 m^4-5 m^3+11 m^2+9 m+2\right)\\
   +4 N \left(4 m^3-2 m^2-4 m-1\right)-2 m \left(m^2+6
   m-9\right)\Big]\\
   +D^4 (m-1)^2 \Big[(N^5 m (m+1)^2 \left(3 m^2-1\right)-N^4 \left(9 m^5+36 m^4+30
   m^3+m+4\right)\\
   +N^3 \left(46 m^4+96 m^3+28 m^2+8 m+22\right)-4 N^2 \left(m^4+22 m^3+18 m^2+11\right)\\
   +8 N \left(m^3+9 m^2-2 m+4\right)-8 m (m+1) \Big]\\
   -D^2 N (m-1) (N m+N-2)^2 \Big[N^4 m \left(3 m^3+5 m^2+3
   m+1\right)-N^3 \left(7 m^4+19 m^3+15 m^2+5 m-2\right)\\
   +4 N^2 \left(4 m^3+6 m^2+3 m-2\right)-2 N \left(m^3+4
   m^2+5 m-8\right)+8 (m-2)\Big]\\
   +(N-2) N^2 (m+1) (N m+N-2)^4 \left(N^2 m^2+N m-2\right)\bigg\}.
\end{multline*}
Above, $N=N_1 + N_2$ is the total population size and $D=|N_1 - N_2|$ is the difference in size between the islands.

For rare migration ($m \to 0$), the critical $b/c$ ratio becomes
\begin{equation}
\lim_{m \to 0} \left(\frac{b}{c} \right)^* = 
\frac{ N^2(N-2)^3+ N \left(N^2-7 N+8\right) D^2+ (2 N-3)D^4}{4 N (N-2) \left(N-D^2\right)}.
\end{equation}
If, in addition, the islands are evenly sized ($D=0$), we have 
\[
\lim_{m \to 0} \left(\frac{b}{c} \right)^* = (N-2)^2/4.
\]
We can show that the above value is at least a local infimum of $(b/c)^*$.  This is based on the following observations (made with the aid of Mathematica):
\renewcommand{\labelenumi}{(\roman{enumi})}
\begin{enumerate}
\item For all $N \geq 4$ and $0<m<1$,
\[
\frac{d(b/c)^*}{dD} \bigg|_{D=0} = 0.
\]
\item For all $N \geq 4$ and $0<m<1$,
\[
\frac{d^2(b/c)^*}{dD^2} \bigg|_{D=0} > 0.
\]
\item For $D=0$ and all $N \geq 4$, $(b/c)^*$ approaches an infimum as $m \to 0$.
\end{enumerate}
We conjecture that $(b/c)^*=(N-2)^2/4$ is a global infimum as well.  

\subsubsection{More than two islands}

For three and four islands, we have obtained the exact critical $b/c$ ratio under the limit $m \to 0$, but the results are too lengthy to record here.  For these cases, we have shown with the aid of Mathematica that $\lim_{m \to 0} (b/c)^*$ is minimized when the islands are evenly sized.  

For five islands, we have obtained $\lim_{m \to 0} (b/c)^*$ under the assumption that two of the islands have equal size.  Under this assumption, we have then shown that $\lim_{m \to 0} (b/c)^*$ is minimized when the remaining three islands also have the same size.

For any number $n$ of evenly-sized islands ($N_i=N/n$), with arbitrary migration, we have
\begin{equation}
\label{eq:nevenislands}
\left(\frac{b}{c} \right)^* = 
\frac{ (N-2)(N-n+mN(n-1))^2}{Nn(N-n+m^2(N-1)) - 2(N-n+mN(n-1))^2}.
\end{equation}
This value of $(b/c)^*$ is minimized as $m \to 0$, at which point it approaches the value
\begin{equation}
\label{eq:evenislandlimit}
\lim_{m \to 0} \left(\frac{b}{c} \right)^* = 
\frac{ (N-2)(N-n)}{Nn-2N+2n}.
\end{equation}
We note that this limiting value is approached in the regime $m \ll 1/n$.  We conjecture that, for each fixed population size $N \geq 4$ and number of islands $n \geq 2$, the right-hand side of Eq.~\eqref{eq:evenislandlimit} is a global infimum of all positive values of $(b/c)^*$ for all migration rates $0<m<1$ and all distributions of island sizes $\{N_i\}_{i=1}^n$, as long as there are at least two individuals per island ($N_i \geq 2$ for each $i$).

\subsection{Joined stars}
\label{sec:joinedstars}

Here we derive the results shown in Fig.~\ref{fig:surgery}bcf of the main text.

\paragraph{Two stars joined by hubs} For two $n$-stars joined by an edge between their hubs, solving Eqs.~\eqref{eq:mrecur} and \eqref{eq:miirecur} yield 
\begin{equation}
\tau_\mathrm{HH}^+ = \frac{8n^3+22n^2+17n+15}{(n+1)^2(2n+5)} \qquad \tau_\mathrm{LL}^+ = \frac{12n+10}{2n+5},
\end{equation}
By inspection, we have
\begin{equation*}
\pi_\mathrm{H} = \frac{n+1}{4n+2}, \qquad \pi_\mathrm{L} = \frac{1}{4n+2}, \qquad p_\mathrm{HH}^{(2)} =  \frac{n^2+n+1}{(n+1)^2}, 
\qquad p_\mathrm{LL}^{(2)} = \frac{1}{n+1}.
\end{equation*}
Substituting these values into Eq.~\eqref{eq:altcritbc} yields
\begin{equation*}
\left( \frac{b}{c} \right)^* = \frac{(n+1)^2 \left(10 n^2+17 n+5\right)}{n(4 n^3+12 n^2+11 n+5)} \xrightarrow{n \to \infty} \frac{5}{2}.
\end{equation*}

\paragraph{Two stars joined leaf to hub} For two $n$-stars with an edge joining the leaf to the hub of the other, a similar procedure yields
\[
\left( \frac{b}{c} \right)^* =\frac{(n+1) \left(36 n^2+90 n+19\right)}{4 n \left(3 n^2+11 n+9\right)} 
\xrightarrow{n \to \infty} 3.
\]

\paragraph{Two stars joined leaf to leaf} For two $n$-stars with an edge joining a leaf of each, we obtain
\[
\left( \frac{b}{c} \right)^* = \frac{2 (n+1) \left(490 n^4+3065 n^3+5982 n^2+4559 n+1136\right)}{350 n^5+2671 n^4+6818
   n^3+7489 n^2+3544 n+568}
   \xrightarrow{n \to \infty} \frac{14}{5}.
\]

\paragraph{``Dense cluster" of stars} For $m$ $n$-stars, with each hub  joined to each other hub, we obtain
\[ \left(\frac{b}{c} \right)^* = \frac{\text{num}}{\text{denom}}\]
where
\begin{multline*}
\text{num} = (m+n-1)^2 \\ 
\times \big(2 m^4+m^3 (11 n-8)+m^2 \left(20 n^2-25n+11\right) \\
+m \left(12 n^3-22 n^2+12 n-7\right)-2 \left(2n^3+n^2+n-1\right)\big),
\end{multline*}
\begin{multline*}
\text{denom}= 2 m^5 (n-1)+m^4 \left(12 n^2-19n+10\right)+m^3 \left(26 n^3-61 n^2+61 n-19\right)\\
+3 m^2 \left(8n^4-28 n^3+40 n^2-29 n+6\right)+m \left(8 n^5-48 n^4+92 n^3-103n^2+57 n-9\right)\\
-8 n^5+24 n^4-34 n^3+32 n^2-14 n+2.
\end{multline*}
Letting $n \to \infty$ for fixed $m$, we obtain 
\[ \lim_{n \to \infty} \left(\frac{b}{c} \right)^* =  \frac{3m-1}{2m-2}.\]
If we then take $m \to \infty$, we have $(b/c)^* \to 3/2$.

\subsection{Ceiling fan}

Let us now consider the ``ceiling fan" graph (Fig.~\ref{fig:surgery}e of the main text), in which each of the $n$ leaves of a star is joined by an edge to one other. Solving Eqs.~\eqref{eq:mrecur} and \eqref{eq:miirecur} yields
\begin{equation*}
\tau_\mathrm{HH}^+ = \frac{9n-3}{n+3}, \qquad \tau_\mathrm{LL}^+ = \frac{15n}{2(n+3)},
\end{equation*}
where H and L indicate hub and leaf vertices, respectively.  By inspection, we have
\begin{equation*}
\pi_\mathrm{H} = \frac{1}{3}, \qquad \pi_\mathrm{L} = \frac{2}{3n}, \qquad p_\mathrm{HH}^{(2)} = \frac{1}{2}, 
\qquad p_\mathrm{LL}^{(2)} = \frac{n+2}{4n}.
\end{equation*}
Substituting these values into Eq.~\eqref{eq:altcritbc} yields
\begin{equation*}
\left( \frac{b}{c} \right)^* = \frac{4(6n-7)}{3n-16} \xrightarrow{n \to \infty} 8.
\end{equation*}

\subsection{Wheel}

In a wheel graph (Fig.~\ref{fig:surgery}e of the main text), each of the $n$ leaves is joined to two neighboring leaves as well as to the hub.  We define $\tau_{L,j}$ to be the coalescence time for two leaves that are $j$ apart, $0 \leq j \leq n$.  Clearly, we have $\tau_{L,0}  = \tau_{L,n} = 0$.  We also define $\tau_{LH}$ to be the coalescence time between a leaf and the hub.

The recurrence relations \eqref{eq:mrecur} for coalescence times become
\begin{align}
\label{eq:tauLj}
\tau_{L,j} & = 1 + \frac{1}{3} \left( \tau_{L,j-1} + \tau_{L,j+1} + \tau_{LH} \right) \quad \text{for $1 \leq j \leq n-1$},\\
\label{eq:tauH1}
\tau_{LH} & = 1 + \frac{1}{3} \tau_{LH} + \frac{1}{2n} \sum_{j=0}^{n-1} \tau_{L,j}.
\end{align}
Solving Eq.~\eqref{eq:tauH1} for $\tau_{LH}$ yields
\begin{equation}
\label{eq:tauH2}
\tau_{LH}  = \frac{3}{2} + \frac{3}{4n} \sum_{j=0}^{n-1} \tau_{L,j}.
\end{equation}
It is convenient to define $\tau_{L,j}'=\tau_{L,j}-\tau_{LH}$.  Then Eqs.~\eqref{eq:tauLj} and \eqref{eq:tauH2} become
\begin{align}
\label{eq:tauLjprime}
\tau_{L,j}' & = 1 + \frac{1}{3} \left( \tau_{L,j-1}' + \tau_{L,j+1}' \right)  \quad \text{for $1 \leq j \leq n-1$},\\
\label{eq:taujsum}
\tau_{LH} & = 6 + \frac{3}{n}\sum_{j=0}^{n-1} \tau_{L,j}'.
\end{align}
We guess a solution (ansatz) of the form
\begin{equation}
\label{eq:wheelansatz}
\tau_{L,j}' = a  + b\left( \gamma^j + \gamma^{n-j} \right).
\end{equation}
Substituting this ansatz into Eq.~\eqref{eq:tauLjprime}, we obtain
\begin{align*}
a  + b\left( \gamma^j + \gamma^{n-j} \right) & = 
1 + \frac{2a}{3} + \frac{b}{3} \left( \gamma^{j-1} + \gamma^{n-j+1} + \gamma^{j+1} + \gamma^{n-j-1} \right)\\
& = 1 + \frac{2a}{3} + \frac{b}{3} \left( \gamma + \gamma^{-1} \right) \left( \gamma^j + \gamma^{n-j} \right).
\end{align*}
For this to hold for all $1 \leq j \leq n-1$ necessitates that
\[
a = 1 + \frac{2a}{3} \quad \text{and} \quad \gamma + \gamma^{-1} = 3,
\]
which gives the solutions
\[
a=3 \quad \text{and} \quad \gamma = \frac{3 \pm \sqrt 5}{2}.
\]
It turns out not to matter which value of $\gamma$ is used; we will use $\gamma = (3-\sqrt{5})/2$.  To solve for $b$, we substitute into Eq.~\eqref{eq:taujsum}, 
\begin{align}
\nonumber
\tau_{LH} & = 6 + \frac{3}{n} \sum_{j=0}^{n-1} \left( 3 + b\left( \gamma^j + \gamma^{n-j} \right) \right)\\
\label{eq:tauLHb}
& = 15 + \frac{3b}{n} \frac{(1+\gamma)(1-\gamma^n)}{1-\gamma}.
\end{align}
Additionally, since $\tau_{L,0}=0$, we have 
\begin{equation}
\label{eq:tauLHb2}
\tau_{LH}  = -\tau_{L,0}' = -b(1+\gamma^n)-3.
\end{equation}
Combining Eqs.~\eqref{eq:tauLHb} and \eqref{eq:tauLHb2} and solving for $b$ yields
\[
b  = \frac{-18n(1-\gamma)}{ 3 (1+\gamma)(1-\gamma^n) +n (1+\gamma^n)(1-\gamma)}.
\]

Substituting this value of $b$ into Eqs.~\eqref{eq:tauLHb2} and \eqref{eq:wheelansatz}, we obtain the coalescence times
\begin{align*}
\tau_{LH} & = \frac{18n(1-\gamma)(1+\gamma^n)}{ 3 (1+\gamma)(1-\gamma^n) +n (1+\gamma^n)(1-\gamma)} - 3,\\
\tau_{L,j} & = \tau_{L,j}' +\tau_{LH} \\
& = \frac{18n(1-\gamma)}{ 3 (1+\gamma)(1-\gamma^n) +n (1+\gamma^n)(1-\gamma)}
\left(1+\gamma^n-\gamma^j - \gamma^{n-j} \right)\\
& =  \frac{18n(1-\gamma)}{ 3 (1+\gamma)(1-\gamma^n) +n (1+\gamma^n)(1-\gamma)}
\left(1-\gamma^j \right)\left(1-\gamma^{n-j}\right).
\end{align*}
In particular, for neighboring leaves ($j=1$), we have
\begin{equation}
\tau_{L,1} = \frac{18n(1-\gamma)^2(1-\gamma^{n-1})}{ 3 (1+\gamma)(1-\gamma^n) +n (1+\gamma^n)(1-\gamma)}.
\end{equation}

Turning now to remeeting times, we compute
\begin{align*}
\tau_{HH}^+ & = 1 + \tau_{LH}\\
& = \frac{18n(1-\gamma)(1+\gamma^n)}{ 3 (1+\gamma)(1-\gamma^n) +n (1+\gamma^n)(1-\gamma)}-2,\\
\tau_{LL}^+ & = 1 + \tfrac{1}{3}\tau_{LH} + \tfrac{2}{3} \tau_{L,1} \\
& = \frac{18n(1-\gamma)}{ 3 (1+\gamma)(1-\gamma^n) +n (1+\gamma^n)(1-\gamma)}
\left(1+\gamma^n - \frac{2\gamma}{3} (1+\gamma^{n-2} ) \right).
\end{align*}
The other values needed to compute $(b/c)^*$ are
\[
\pi_H = \frac{1}{4}, \qquad \pi_L = \frac{3}{4n}, \qquad p_{HH}^{(2)} = \frac{1}{3}, \qquad p_{LL}^{(2)} = \frac{2n+3}{9n}.
\]
Using the above values, the critical $b/c$ ratio can be obtained from Eq.~\eqref{eq:altcritbc}
\begin{equation}
\label{eq:bcwheel}
\left(\frac{b}{c} \right)^* 
 = \frac{\pi_H \tau_{HH}^+  + n  \pi_L \tau_{LL}^+ - 2}
{\pi_H \tau_{HH}^+ p_{HH}^{(2)}  + n  \pi_L \tau_{LL}^+ p_{LL}^{(2)} - 2}.
\end{equation}

Now turning to the $n \to \infty$ limit, we calculate:
\begin{align*}
\lim_{n \to \infty} \tau_{HH}^+ & = 16,\\
\lim_{n \to \infty} \tau_{LL}^+ & = 18 - 12 \gamma=6\sqrt{5},\\
\lim_{n \to \infty} p_{LL}^{(2)} & = \frac{2}{9}.
\end{align*}
Substituting into Eq.~\eqref{eq:bcwheel} and simplifying gives
\[
\lim_{n \to \infty} \left(\frac{b}{c} \right)^* = \frac{429 + 90\sqrt{5}}{82}.
\]

\section{Direct and inclusive fitness}
\label{sec:DirectInclusive}

The conditions for success derived in Sections \ref{sec:fixweak}--\ref{sec:examples} are based on fixation probability and, in the case of nonzero mutation, expected degree-weighted abundance.  Other approaches in the literature are based on the fitness and/or inclusive fitness of individuals.  In the interest of synthesizing different approaches, we calculate the fitness and---in the case of the donation game \eqref{eq:PD}---the inclusive fitness effect associated to each vertex.

\subsection{Fitness}

The (direct) fitness of an individual is a measure of its reproductive success. In homogeneous populations, the fitness of an individual is defined as its survival probability plus its expected number of offspring.  For heterogeneous populations, individuals have different \emph{reproductive values}---that is, they make different expected contributions to the future gene pool of the population, even under neutral drift \cite{taylor1990allele,maciejewski2014reproductive,taylor2014hamilton,tarnita2014measures}.  For our model, we identify the reproductive value of vertex $i$ as its relative weighted degree $\pi_i$ \cite{maciejewski2014reproductive}, which is also equal to the fixation probability of a neutral mutation arising at this vertex \cite{maciejewski2014reproductive,allen2015molecular}.

We formally define the fitness of an individual as its survival probability multiplied by its own reproductive value, plus the expected total reproductive value of all offspring it produces, over a short time interval $[t, t+\epsilon)$. The fitness $v_i(\vs)$ of vertex $i$ in state $\vs$ is calculated as follows:
\begin{align}
\nonumber
v_i(\vs) & = \pi_i \left( 1- \epsilon \sum_{j \in G} \operatorname{Rate}[j \rightarrow i](\vs) \right)
+ \epsilon \sum_{j \in G}  \pi_j \operatorname{Rate}[i \rightarrow j](\vs) + o(\epsilon)\\
\nonumber
& = \pi_i + \epsilon \left( \sum_{j \in G} \frac{ \pi_j w_{ij} F_i(\vs)}{\sum_{k \in G} w_{kj} F_k(\vs)} - \pi_i \right) + o(\epsilon)\\
\label{eq:directfit}
& = \pi_i + \epsilon \delta \pi_i \left( f_i(\vs) -  f_j^{(2)} (\vs) \right) + R(\epsilon,\delta),
\end{align}
where the remainder term $R(\epsilon, \delta)$ satisfies 
\[
\lim_{\epsilon \to 0^+} \lim_{\delta \to 0} \frac{R(\epsilon,\delta)}{\epsilon\delta} = 0.
\]

The constant term in Eq.~\eqref{eq:directfit} is the fitness of vertex $i$ under neutral drift, which is equal to  its reproductive value $\pi_i$.  The second term represents the effects of weak selection.  We define the \emph{direct fitness effect} of selection on individual $i$, denoted $v_i'(\vs)$ to be the coefficient of $\epsilon \delta$:
\begin{equation}
\label{eq:directfiteffect}
v_i'(\vs) = \pi_i\left( f_i(\vs) -  f_i^{(2)} (\vs) \right).
\end{equation}
For an arbitrary matrix game \eqref{eq:game}, substituting the payoffs from Eq.~\eqref{eq:fgeneral}, we obtain
\begin{equation}
\begin{split}
\label{eq:directfiteffectgeneral}
v_i'(\vs) & = \pi_i\Bigg( a \left[ s_i s_i^{(1)} - \sum_{j \in G} p_{ij}^{(2)} s_j s_j^{(1)}  \right]\\ 
& \qquad + b \left[ s_i \left(1-s_i^{(1)} \right) - \sum_{j \in G} p_{ij}^{(2)} s_j   \left(1-s_j^{(1)} \right) \right] \\
& \qquad + c \left[ (1-s_i) s_i^{(1)} - \sum_{j \in G} p_{ij}^{(2)} \left(1-s_j \right) s_j^{(1)} \right]   \\
 & \qquad + d \left[  (1-s_i)  \left(1-s_i^{(1)} \right)  
 - \sum_{j \in G} p_{ij}^{(2)}\left(1-s_j \right) \left(1-s_j^{(1)} \right) \right] \Bigg).
\end{split}
\end{equation}

For the donation game \eqref{eq:PD}, Eq.~\eqref{eq:directfiteffectgeneral} simplifies to
\begin{equation}
\label{eq:directfiteffectPD}
v_i'(\vs) = \pi_i \left(-cs_i + b s_i^{(1)} + cs_i^{(2)} - b s_i^{(3)} \right).
\end{equation}

Eqs.~\eqref{eq:directfiteffect}--\eqref{eq:directfiteffectPD} apply to a particular state $\vs$.  We can also calculate  the direct fitness effect of a particular strategy at a particular vertex for the overall evolutionary process.  For this we introduce mutation with probability $0<u<1$, as discussed in Section \ref{sec:mutation}.  The expected direct fitness effect of strategy A at vertex $i$, under the neutral MSS distribution, can be written as 
\begin{equation}
\label{eq:directMSS}
\E^\circ_\mathrm{MSS} \left[ v_i' | s_i=1  \right ]  =  \pi_i \E^\circ_\mathrm{MSS} \left[ f_i-  f_i^{(2)}  \big |  s_i=1 \right].
\end{equation}

To explicitly compute the right-hand side of \eqref{eq:directMSS} for the general game \eqref{eq:game} requires triplet IBD probabilities \cite{taylor2013inclusive,taylor2016hamilton} and is beyond the scope of this work.  However, for the donation game \eqref{eq:PD}, only pairwise IBD probabilities are required.  A generalization of Lemma 3 of \cite{allen2014games} yields
\[
\E^\circ_\mathrm{MSS} \left[ s_j | s_i=1  \right ]  =  \frac{1+q_{ij}}{2}.
\]
Applying this identity to Eq.~\eqref{eq:directfiteffectPD} yields the overall direct fitness effect of cooperation at vertex $i$:
\begin{align*}
\E^\circ_\mathrm{MSS} \left[  v_i' | s_i=1  \right ] & = \frac{\pi_i}{2} \left( -c\left( q_i^{(0)} - q_i^{(2)} \right) + b \left( q_i^{(1)} - q_i^{(3)}\right) \right)\\
& = \frac{u \pi_i}{2} \left[ -c \tau_i^{(2)} + b\left( \tau_i^{(3)} -\tau_i^{(1)} \right) \right] + \mathcal{O}(u^2)
\qquad \qquad  (u \to 0).
\end{align*}
Above, we have defined $q_i^{(n)} = \sum_{j \in G} p_{ij}^{(n)} q_{ij}$ and $\tau_i^{(n)} = \sum_{j \in G} p_{ij}^{(n)} \tau_{ij}$.  The second equality above uses Eq.~\eqref{eq:qtau}.

\subsection{Inclusive fitness}

Inclusive fitness theory \cite{Hamiltona,RoussetBook,Taylor} analyzes the evolution of social behavior using a quantity called the \emph{inclusive fitness effect}.  This quantity is defined as the effect that this individual has on its own fitness, plus a weighted sum of the effects it has on the fitnesses of all others, where the weights quantify genetic relatedness.  

In order to formulate such a quantity, there must be a well-defined contribution that each individual makes to its own fitness and to the fitness of each other individual.  However, Eq.~\eqref{eq:directfiteffectgeneral} for the direct fitness effect is quadratic in $s_1, \ldots, s_N$, and does not separate into distinct contributions due to particular individuals.  Therefore, for the general game \eqref{eq:game}, there is no inclusive fitness effect of a single given individual.  (But see Refs.~\cite{taylor2013inclusive,taylor2016hamilton,akccay2016there} for alternative notions of inclusive fitness at the level of pairs or genetic lineages.)

For the donation game \eqref{eq:PD} the direct fitness effect \eqref{eq:directfiteffectPD} is linear in $s_1, \ldots, s_N$ and does separate into distinct contributions due to particular individuals.  Specifically, we can write 
\begin{equation}
\label{eq:modulated}
v_i'(\vs) = \sum_{j \in G} e_{ji} s_j,
\end{equation}
where we have defined the \emph{fitness effect of cooperation at $j$ on individual $i$} as 
\begin{equation}
\label{eq:fiteffect}
e_{ji} = \pi_i \left( -cp_{ij}^{(0)} + b p_{ij}^{(1)} + cp_{ij}^{(2)} - b p_{ij}^{(3)} \right).
\end{equation}
Because the fitness effects $e_{ij}$ are well-defined for game \eqref{eq:PD}, the inclusive fitness effect of cooperation at vertex $i$ exists and can be written as 
\begin{equation}
\label{eq:incfiteffect}
v_i^\mathrm{IF}(\vs) = \sum_{j \in G} e_{ij} s_j.
\end{equation}
In the above expression, the relatedness of a hypothetical cooperator at vertex $i$ to the occupant of vertex $j$ is defined to be $s_j$.  This notion of relatedness  applies to a particular state $\vs$, and might be termed ``identity in state".  It is common to normalize relatedness coefficients so that they lie in the range -1 to 1; for example, one might use $(s_j-\bar{s})/(1-\bar{s})$ instead of $s_j$ to quantify the relatedness of a cooperator at $i$ to the occupant of $j$, where $\bar{s} = \frac{1}{N} \sum_{k \in G} s_k$ is the average population type. However, such normalizations are not needed for models with constant population size \cite{Taylor}.   

We observe that, as a consequence of the reversibility property $\pi_i p_{ij}^{(n)} = \pi_j p_{ij}^{(n)}$, fitness effects are symmetric: $e_{ij} = e_{ji}$.  In other words, the effect that cooperation at $i$ has on vertex $j$ is equal to the effect that cooperation at $j$ has on vertex $i$.  It follows that the direct and inclusive fitness effects of cooperation at $i$ are equal in every state: $v_i^\mathrm{IF}(\vs) = v_i'(\vs)$.  This is an interesting but idiosyncratic property of the model we consider.  We would not, for example, find the result for different interaction and replacement graphs (Section \ref{sec:interactionreplacement}) since there is no analogous reversibility property for $(n,m)-$random walks.  

Since---for the donation game \eqref{eq:PD}---the direct and inclusive fitness effects of cooperation at $i$ are equal in every state, they are also equal for the overall evolutionary process:
\begin{align*}
\E^\circ_\mathrm{MSS} \left[  v_i^\mathrm{IF} | s_i=1  \right ] & =
\E^\circ_\mathrm{MSS} \left[  v_i' | s_i=1  \right ]\\ 
& = \frac{\pi_i}{2} \left( -c\left( q_i^{(0)} - q_i^{(2)} \right) + b \left( q_i^{(1)} - q_i^{(3)}\right) \right)\\
& = \frac{u \pi_i}{2} \left[ -c \tau_i^{(2)} + b\left( \tau_i^{(3)} -\tau_i^{(1)} \right) \right] + \mathcal{O}(u^2)
\qquad \qquad  (u \to 0).
\end{align*}

\section{Computational issues}
\label{sec:computational}

As we pointed out below Eq.~\eqref{eq:mrecur}, computing the coalescence times involves solving a system of $\binom{N}{2}$ linear equations. Simple Gaussian elimination takes $\mathcal{O}(N^6)$ steps for such a system. However this can be improved by a (standard) block-wise inversion approach combined with a state-of-the-art matrix multiplication algorithm. For example, based on variants of the Coppersmith-Winograd algorithm, coalescence times can be computed in $\mathcal{O}(N^{4.75})$ time. 

Further improvement can be achieved by allowing approximate solutions and observing that \eqref{eq:mrecur} is a symmetric diagonally dominant (SDD) system. Such systems can be solved in nearly linear time in the number of non-zero entries of the coefficient matrix. More precisely, for a system described by an $n \times n$ matrix with $m$ non-zero entries, finding a vector $\epsilon$ far in norm from the exact solution can be done in $\mathcal{O}(m \log^2 n \log 1/\epsilon)$ time \cite{KMP2010}.  It follows that for a graph of size $N$ with average degree $\bar{k}$, coalescence times can be determined to within $\epsilon$ in $\mathcal{O}\big( (N\log N)^2\, \bar{k} \log 1/\epsilon \big)$ time.  Furthermore, the algorithm can be efficiently parallelized \cite{KLPSS2015}.

In our experiments we have used a MatLab implementation by Koutis (\url{http://tiny.cc/cmgSolver}). Representative run times on a 2011 MacBook Air, MatLab 2015a were as follows: for $N=1000$, average degree 4, the running time was 12 seconds. For $N=2000$, average degree 4, the running time was 120 seconds. For $N=2000$, average degree 8, the running time was 280 seconds. 

For Figure \ref{fig:million} we computed $(b/c)^*$ for 1.3 million unweighted graphs, generated from 10 different random graph models.  Parameter values were sampled from a uniform distribution on the specified ranges (see below).  Initial graph sizes were uniformly sampled in the range $100 \leq N \leq 150$; if the random graph model produced a disconnected graph, the largest connected component was used.  The critical $(b/c)^*$ ratio was computed by solving Eq.~\eqref{eq:taurecur} for coalescence times and substituting into Eq.~\eqref{eq:altcritbc}.  (No Monte Carlo simulations were used for these investigations.)

Random graph models and parameter ranges are as follows: 100K Erdos-Renyi \cite{ER} with edge 
 probability $0 < p < 1$; 100K small world \cite{newman} with initial connection distance $1\leq d \leq 5$ and edge creation probability $0<p<0.05$; 100K Barabasi-Albert \cite{barabasi} with linking number $1 \leq m \leq 10$; 100K random recursive \cite{RRT} (like Barabasi-Albert except that edges are added uniformly instead of preferentially) with linking number $2 \leq m \leq 8$, 200K Holme-Kim \cite{HK} with linking number $2\leq m \leq 4$ and triad formation parameter $0< P < 0.15$; 200K Klemm-Eguiluz \cite{KE} with linking number $3\leq m \leq 5$ and deactivation parameter $0< \mu <0.15$; 200K shifted-linear preferential attachment \cite{shifted} with linking number $1\leq m \leq 7$ and shift $0< \theta <40$; 100K forest fire \cite{FF} with parameters $0 < p_f < p_b < 0.15$; 100K Island BA; and 100K Island ER.  Island BA is a meta-network of islands \cite{island}, in which each island is a shifted-linear preferential attachment with the same parameters as above.  The number of islands varies from 2 to 5. Considering the islands as meta-nodes, the meta-network among the islands is an ER graph with edge probability $0< p_{\mathrm{inter}}<1$.  Island ER is the as Island BA except that each island is an ER graph with edge probability $0<p_\mathrm{intra}<1$.

We also computed the critical $b/c$ ratio for some large real-world networks using Northeastern's computational cluster. This was computed on a node with Intel Xeon CPU E5-2680 2.8GHz and 256GB RAM.
\begin{itemize}
\item The Framingham Study graph $(N=5253$, average degree $6.5)$. This network has a critical $b/c$ ratio of $7.96$.  The running time was around 3.5 hours. 
\item The ego-facebook network from the Stanford SNAP database ($N=4039$,  average degree $43.7$) took 25 minutes and has critical $b/c$ ratio $48.5$.
\item The (largest connected component of the) ca-GrQc graph from the Stanford SNAP database ($N=4159$, average degree $6.5$ took 23 hours and has critical $b/c$ ratio $6.6$. The reason for this is most likely that this graph is very badly conditioned.
\end{itemize}

\section{Monte Carlo Simulations}

\begin{figure}
        \centering
        \begin{subfigure}[b]{0.45\textwidth}
                \includegraphics[width=.85 \textwidth]{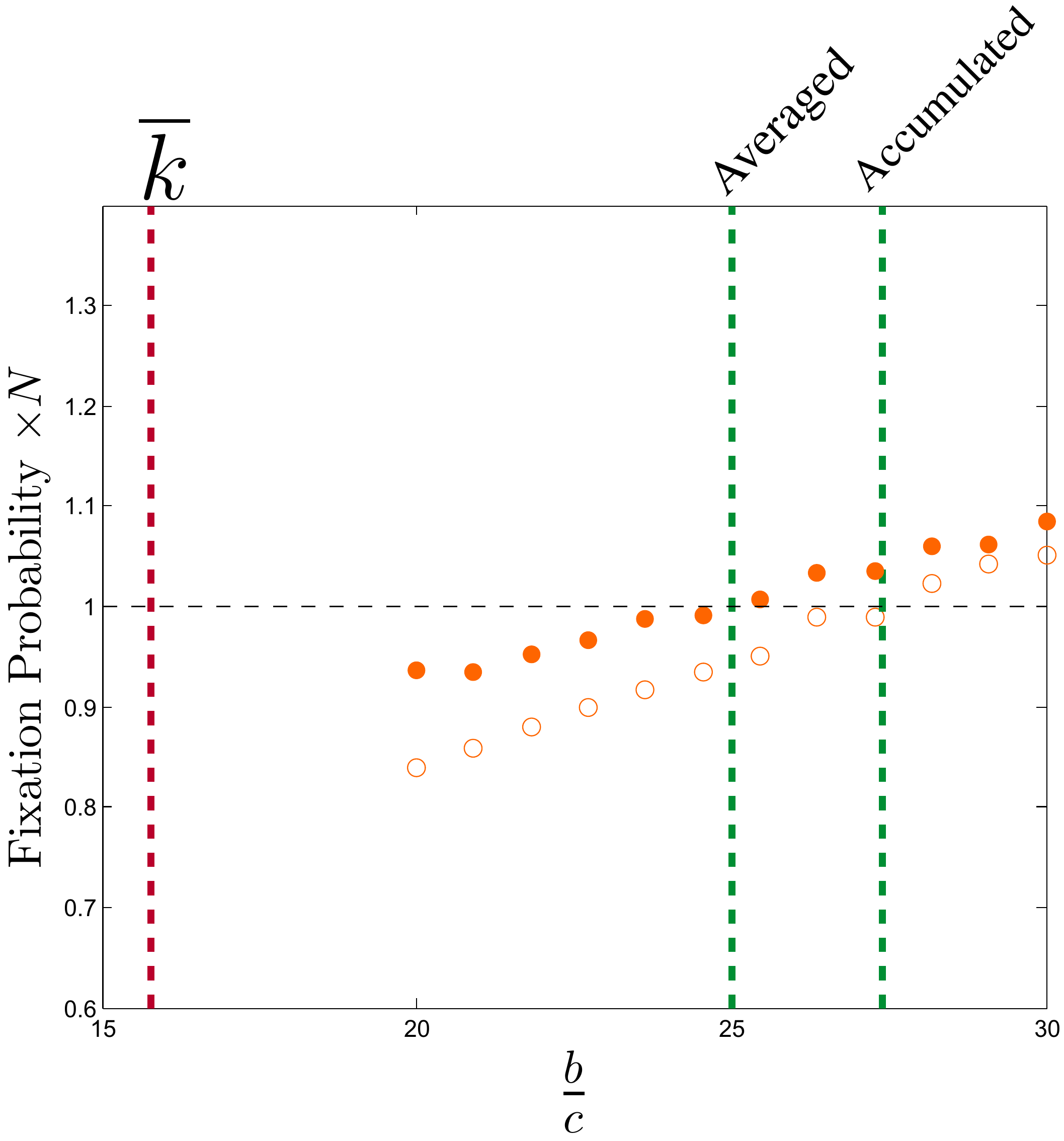}
                \caption{
BA
}
                \label{BA}
        \end{subfigure}%
        ~ ~~ 
        \begin{subfigure}[b]{0.45\textwidth}
                \includegraphics[width=.85  \textwidth]{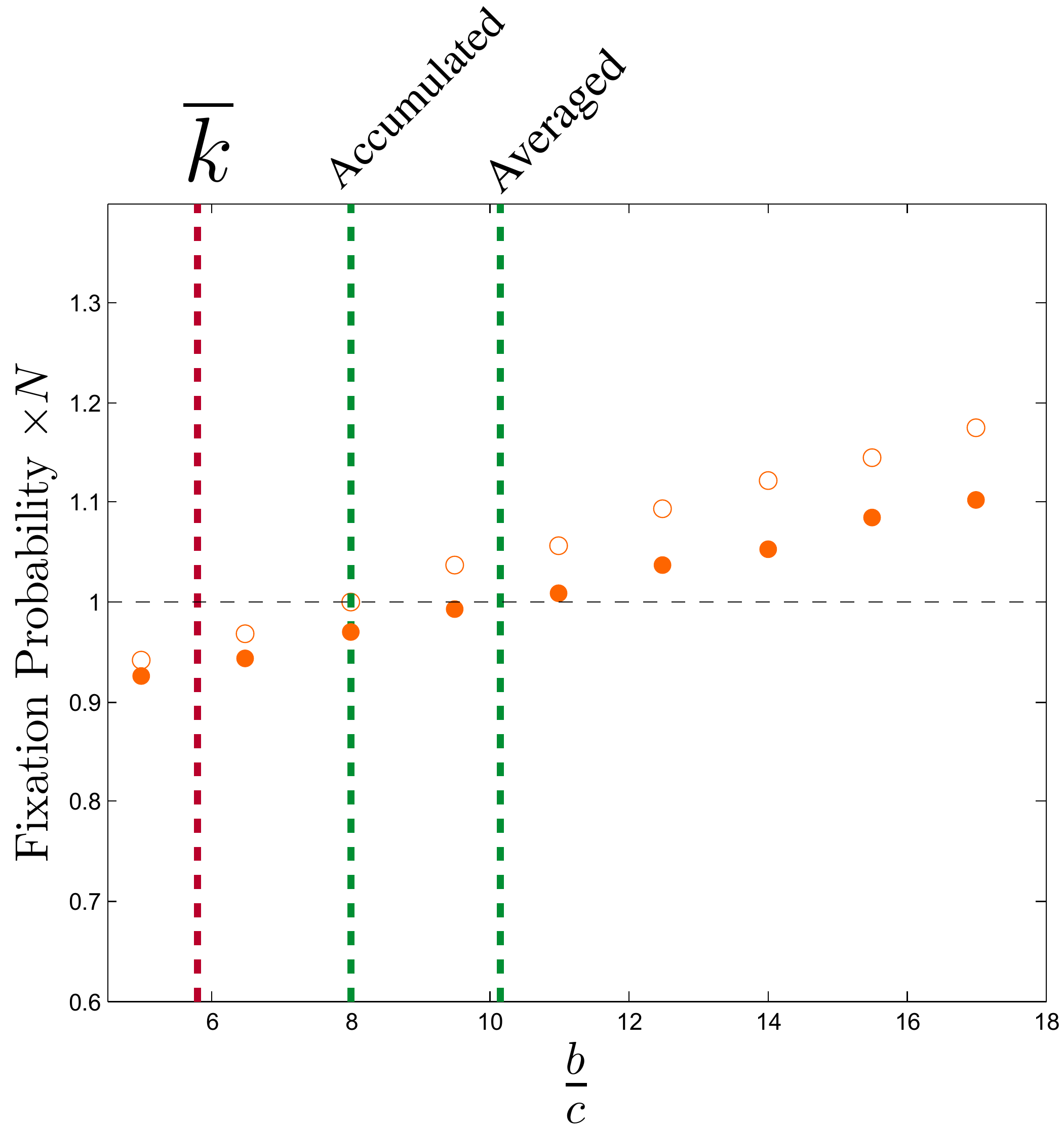}
                \caption{KE}
                \label{KE2}
        \end{subfigure}
\label{b8_k6_L12_d2_N200_T250_pSW0}

 \begin{subfigure}[b]{0.45\textwidth}
                \includegraphics[width=.85 \textwidth]{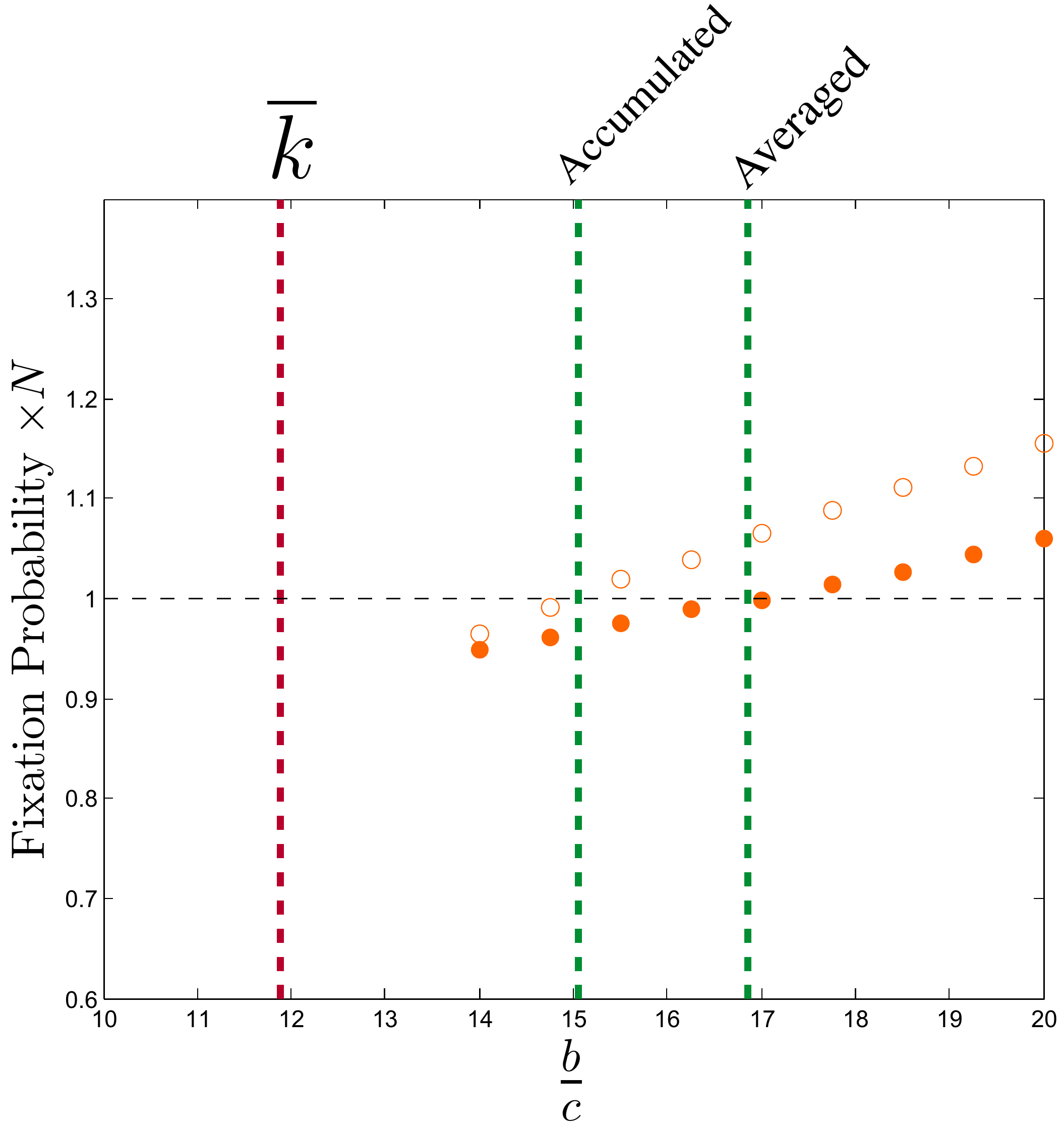}
                \caption{
SW
}
                \label{SW}
        \end{subfigure}%
        ~ ~~ 
        \begin{subfigure}[b]{0.45\textwidth}
                \includegraphics[width=.85  \textwidth]{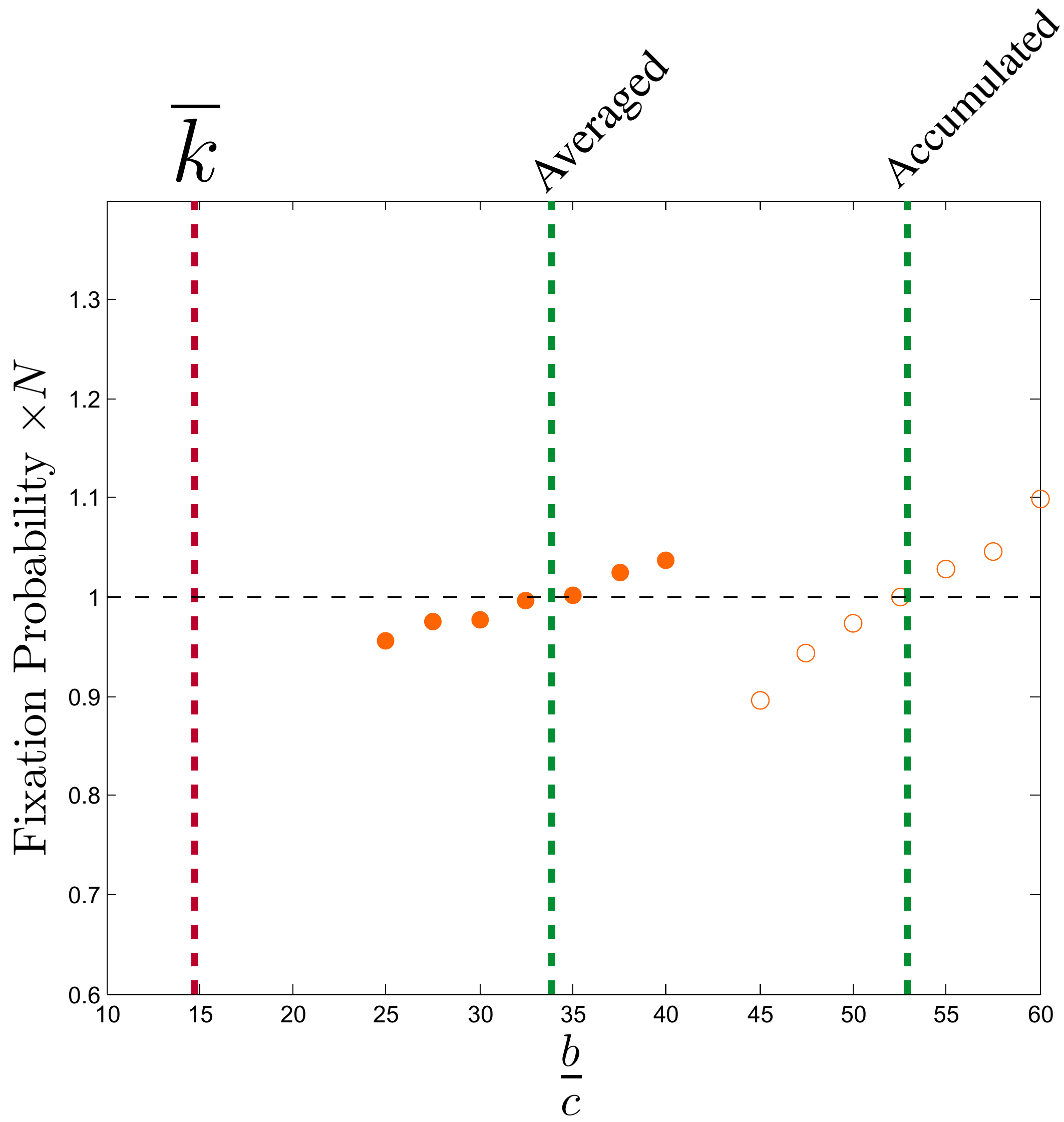}
                \caption{
                FF}
                \label{FF2}
        \end{subfigure}
         \caption{\textbf{Results from additional Monte Carlo simulations.} Simulation are shown for averaged payoffs (solid dots) and accumulated payoffs (open dots).  Vertical lines indicate theoretical $(b/c)^*$ for averaged and accumulated payoffs, as well as the mean degree $\bar{k}$. The horizontal line pertains to neutral drift, for which ${N\rho=1}$. (a) Barabasi-Albert network~\cite{barabasi} with $m=4, N=100$. (b) Klemm-Eguiluz network \cite{KE} with $m=3,\mu=0.4, N=100$.  (c) Small world network \cite{newman} with initial connection distance $d=3$, link creation probability $p=0.1$ and $N=80$.  (d) Forest fire network~\cite{FF} with parameters $p_b=0.32, p_f=0.28, N=80$.}
\label{FIG_ALL_1}
\end{figure}

To verify that our results (which are exact in the limit of weak selection) are accurate for nonweak selection, we performed Monte Carlo simulations. Results are presented in Figures \ref{fig:simulation} and \ref{FIG_ALL_1}.   

The simulation setup is as follows. For each graph, we run $5 \times 10^5$ Monte Carlo trials. For each trial, all nodes are defectors upon the inception, except one randomly-selected node which is a cooperator. The fixation probability is approximated as the fraction of Monte Carlo trials which eventuate in unanimous cooperation before timestep  $T$, which is set to 400000 (which theoretically should be infinite).  The cost of cooperation is $c=1$.  For Fig.~\ref{fig:simulation}, selection strength is $\delta = 0.025.$  For Fig.~\ref{FIG_ALL_1}, selection strength is $\delta = 10^{-2}$ for averaged payoffs and $\delta = 10^{-3}$ for accumulated payoffs.  The smaller $\delta$-value for accumulated payoffs is needed to compensate for the summing of a potentially large number of individual payoffs.

\bibliographystyle{unsrt}
\bibliography{hetgraph}

\end{document}